\documentclass[11pt]{article}
\usepackage[margin=1in]{geometry}
\usepackage{authblk}

\usepackage{amsfonts}
\usepackage{amssymb}
\usepackage[title]{appendix}
\usepackage{booktabs}
\usepackage{dsfont}
\usepackage{xcolor}

\usepackage{amsmath}
\usepackage{mathtools}
\DeclareMathOperator{\dist}{dist}
\DeclareMathOperator{\PP}{\mathbb{P}}
\DeclareMathOperator{\E}{\mathbb{E}}
\DeclareMathOperator{\var}{\mathbb{V}ar}
\DeclareMathOperator{\cov}{\mathbb{C}ov}

\usepackage{amsthm}
\newtheorem{corollary}{Corollary}
\newtheorem{lemma}{Lemma}

\newtheorem{theorem}{Theorem}
\theoremstyle{definition}
\newtheorem{assumption}{Assumption}
\theoremstyle{remark}

\usepackage[natbib=true, style=authoryear-comp, backend=biber]{biblatex}
\PassOptionsToPackage{hyphens}{url}
\usepackage{hyperref}
\hypersetup{
    colorlinks=true,
    linkcolor=red!75!black,
    citecolor=purple!70!black,
    filecolor=magenta!75!black,
    urlcolor=blue!75!black,
}

\usepackage{cleveref}
\crefname{assumption}{assumption}{assumptions}
\Crefname{assumption}{Assumption}{Assumptions}
\crefname{appendix}{appendix}{appendices}
\Crefname{appendix}{Appendix}{Appendices}
\creflabelformat{equation}{#2#1#3}

\usepackage{subcaption}
\usepackage{tikz}
\usepackage{graphicx}
\graphicspath{{plots/}}

\title{Network Experiments with Edge Treatments and Node Outcomes%
\thanks{We thank Peter Tu, Dominic Coey, Chenkun Wang, and Tuan Phan for helpful discussions and comments. All errors are our own.}}
\author[1]{Artem Kuriksha}
\author[2]{Kenneth Hung}
\affil[1]{Ranking AI Research, Meta}
\affil[2]{Ads Online Experimentation, Meta}
\affil[ ]{\texttt{\char`\{}\href{mailto:kuriksha@meta.com}{\textcolor{black}{\texttt{kuriksha}}}\texttt{, }\href{mailto:kenhung@meta.com}{\textcolor{black}{\texttt{kenhung}}}\texttt{\char`\}@meta.com}}
\date{October 5, 2026}

\newcommand{\Bern}{\mathrm{Bern}}
\newcommand{\toprob}{\stackrel{p}{\to}}
\newcommand{\todist}{\stackrel{d}{\to}}
\newcommand{\dmax}{d_\mathrm{max}}
\newcommand{\Normal}{\mathcal{N}}
\newcommand{\thetan}{\theta^{(n)}}
\newcommand{\thetahat}{\hat\theta^{(n)}}
\newcommand{\DeltaY}{\Delta_Y}

\newcommand{\DeltaHt}{\widetilde{\Delta}_H}
\newcommand{\abarn}{\bar{a}^{(n)}}
\newcommand{\Xbarn}{\bar{X}_1^{(n)}}
\newcommand{\phatn}{\hat{p}_1^{(n)}}
\newcommand{\sigman}{\sigma_n}
\newcommand{\sigmanhat}{\hat{\sigma}_n}
\newcommand{\sigmandiag}{\hat{\sigma}_{n,\mathrm{diag}}}
\newcommand{\sigmantilde}{\tilde{\sigma}_n}

\begin{document}

\maketitle

\begin{abstract}
    We present a methodology for analyzing node-level outcomes while experimenting with edge-level treatments in a population connected by an undirected graph. Under our design, nodes are randomly assigned to test or control, and each edge inherits the treatment of its endpoints, with conflicts resolved by randomization. We use each node's assigned status as an instrument for its treatment exposure. We show that the Wald estimator is consistent for the global average treatment effect (GATE), even when the edge weights used to construct the exposure are misspecified. We formalize the assumptions needed both in terms of the linearity of potential outcomes and the sparsity of the graph, and prove the asymptotic normality of the Wald estimator. The estimator is straightforward to implement, requiring no assignment simulations over the graph. Our Monte Carlo study demonstrates the strong performance of our approach in the context of a social platform.
\end{abstract}

\section{Introduction}
\label{sec:intro}

Many online experiments involve treatments that are naturally defined at the edge level of a social network. For instance, a messaging or calling platform may test a new feature in the communication channel between pairs of users, or a social network may test changes to the algorithm governing content shared between friends. While the treatment is inherently an edge-level intervention, the outcomes of interest, such as user engagement or satisfaction, are measured at the node level. This mismatch creates a practical tension: a control node and a test node sharing an edge cannot both have a consistent experience, since their shared edge can only carry one version of the treatment. The fundamental challenge is estimating what would happen if the treatment were applied to \emph{all} edges, given that any finite experiment can only observe a mix of treated and untreated edges for most nodes.

For edge-level treatments, one natural approach is edge-level randomization, in which each edge independently receives a treatment assignment. However, a simple edge-level analysis cannot recover effects on node-level outcomes. Moreover, even for outcomes also recorded at the edge level, spillovers between edges incident to the same node pose a serious threat to its validity, since a user's behavior on one channel likely depends on the treatment of their other channels. \citet{harshaw2023design} enable extrapolation in this case by assuming a linear exposure-response model and using a Horvitz--Thompson-style estimator. However, this design suffers from the problem that the exposure of a high-degree node concentrates around its mean, making extrapolation to the global counterfactual difficult. As a result, the estimator will have high variance.

Another potential solution is to randomize at the cluster level, assigning entire groups of connected users to the same arm. This delivers a consistent experience along within-cluster edges, but it can reduce statistical power, since the effective sample size shrinks toward the number of clusters.

We argue that node-level randomization offers a promising alternative. By assigning a treatment label to each node and letting the treatment of an edge be determined by the labels of its endpoints, we ensure that most edges incident to a node share the same treatment. From the perspective of any given node, the node-level assignment acts as an \emph{instrumental variable} (IV) for the node's aggregate exposure to treated edges: it is randomly assigned, affects the outcome only through the exposure, and shifts the exposure toward the treated or control end of the spectrum.

In this paper, we formalize this approach and propose the Wald estimator for GATE. The estimator is simple to implement and accommodates many practical features of online experiments, such as multiple arms or holdout groups. It does not require network simulations, which are difficult to run and sometimes infeasible when the graph information is not available after the experiment. Moreover, even when the weights the researcher uses to construct the exposure are misspecified, the Wald estimator remains consistent for GATE, so the true edge weights need not be known. Crucially, our methodology delivers variance close to that of a standard node-level A/B test.

Our main contributions are fourfold. First, we present a node-level randomization design and show that it induces \emph{uniform compliance} (\Cref{lma:uniform-compliance}). This means that the response of expected exposure to status assignment is identical across all units, which ensures that the Wald estimator targets GATE rather than a local average treatment effect (LATE) weighted by individual responsiveness. Second, we prove the consistency of the Wald estimator under a linear potential outcomes model and sparse graph assumptions (\Cref{thm:consistency}), and characterize its limit under deviations from linearity (\Cref{thm:nonlinear-weighting}). Third, we establish asymptotic normality (\Cref{thm:normal}) and propose both a conservative variance estimator and a simpler approximate estimator, deriving their asymptotic behavior (\Cref{thm:variance-limit,thm:diagonal-quadratic}). Fourth, we validate the performance of the Wald estimator and both variance estimators in an empirical setting.

Our analysis has three main limitations. First, we restrict the exposure mapping so that each node's outcome depends only on its own exposure, ruling out interference between distinct nodes. Interference between edges incident to a single node remains fully accommodated through the exposure aggregate. Second, identifying GATE requires linearity in exposure. Without it, Wald instead targets the weighted average response characterized in \Cref{thm:nonlinear-weighting}. Finally, we treat the graph as exogenous to the treatment, so settings in which the treatment itself shifts the edge weights fall outside our analysis.

% Bipartite and edge-level experiments (harshaw2023design)
Our work is connected to several strands of the network experiments literature, most directly to bipartite causal inference \citep{zigler2020bipartite, papadogeorgou2025causal}. There, intervention units and outcome units belong to distinct populations. Randomization over intervention units corresponds to edge-level randomization in our setting. The work most directly comparable to ours is \citet{harshaw2023design}, who share both our linearity assumption and our estimand (GATE), and whose identification relies on an exposure mapping that aggregates edge-level assignments into a treatment intensity for the outcome unit. Their exposure-reweighted linear (ERL) estimator applies inverse-variance reweighting. The problem is that, for high-degree nodes, the exposure is an average over many edges and concentrates around its mean, driving the exposure variance in ERL's denominator toward zero and inflating the inverse-variance weights and estimator variance. To counteract this, \citet{harshaw2023design} build on \citet{pouget2019variance} and pair ERL with a clustering step that correlates assignments to restore exposure variance. Under our node-level design, the exposure is largely driven by the node's own treatment assignment, creating substantially more variation in exposure. Our design makes the Wald estimator consistent. At the same time, the estimator is substantially easier to deploy than ERL and more robust to the practical complexities of online experiments.

% Design of network experiments (cluster randomization: ugander2013graph, eckles2017design)
Within the broader literature on the design of network experiments, much of the work focuses on cluster randomization \citep{gui2015network, karrer2021network, eckles2017design}, where the central motivation is reducing interference bias. Cluster-randomized experiments partition the network into clusters and randomize treatment at the cluster level, ensuring that within-cluster edges receive the same treatment. Therefore, clustering can also serve as a way to deliver a consistent user experience along most edges. Our approach takes a fundamentally different perspective: we use node-level randomization and the IV framework to correct for the resulting dilution. This allows for a much larger effective sample size and eliminates the need for a clustering step, which requires knowledge of the graph structure at design time and is sensitive to the quality of the partition. The cost is that we assume away spillovers from neighbors' exposures to a node's outcome, whereas cluster designs aim to contain interference physically through the clustering.

% Instrumental variables with variable treatment intensity (wald1940fitting, angrist1995twostage, imbens1994identification, angrist2000interpretation)
Our use of the Wald estimator connects to the classical literature on instrumental variables. \citet{imbens1994identification, angrist1996identification} establish the LATE theorem: under monotonicity and a nonzero first stage, the Wald estimator with a binary instrument identifies the average treatment effect for compliers. \citet{angrist1995twostage, angrist2000interpretation} extend this to settings with variable treatment intensity, where the IV estimator converges to a weighted average with weights that depend on both the outcome's sensitivity to treatment and the treatment's responsiveness to the instrument. In our setting, the linearity of potential outcomes and the uniformity of compliance eliminate this heterogeneity and ensure that the Wald estimator recovers GATE. Our technical contribution is establishing consistency and asymptotic normality of the Wald estimator in the network setting, and deriving the sparsity conditions on the graph topology under which these results hold.

% IV under interference and network exposure (eckles2016estimating, kang2016peer, borusyak2020nonrandom, adao2019shift)
Similarly to us, several literatures apply IV to settings where exposure is a non-randomized function of multiple randomized assignments. \citet{eckles2016estimating} randomize encouragements to a focal user's peers and instrument for the resulting volume of feedback received by the focal user. \citet{kang2016peer} randomize encouragement at both the cluster and individual level under partial interference, instrumenting for own and peer treatment uptake. \citet{borusyak2020nonrandom} formalize IV with a known exposure mapping applied to randomized shocks, and recenter the instrument to correct for non-random variation in the mapping across units. These papers target identification under different scenarios of endogeneity. Our work targets a framework for network experiments at scale, with identification of GATE delivered by the experiment design and a tuning-free Wald estimator.

The remainder of this paper is organized as follows. \Cref{sec:setting} introduces the setting and the proposed assignment rule, and \Cref{sec:estimator} presents the proposed estimator. \Cref{sec:consistency} establishes consistency, and \Cref{sec:normality} proves asymptotic normality and variance estimation. \Cref{sec:simulations} studies performance in simulations based on production data from interactions on a social platform. \Cref{sec:practical} discusses practical extensions of the framework. \Cref{sec:conclusion} concludes.

\section{Setting}
\label{sec:setting}

We have a population of $n$ nodes connected by an undirected graph $G$, and assume that every node has at least one edge. The graph is equipped with a nonnegative weighted adjacency matrix $W = (w_{ij})_{i, j}$, with $w_{ij} = 0$ if $i$ and $j$ are not connected in $G$ and $\sum_j w_{ij} = 1$ for all nodes $i$. We allow $W$ to be asymmetric: $w_{ij}$ need not equal $w_{ji}$. For example, nodes may represent users on a messaging platform, with an edge connecting two users who have communicated. The weight $w_{ij}$ could be the share of user $i$'s total messages exchanged with user $j$. We define the $k$-hop neighborhood of node $i$ as $N_k(i) = \{j: \dist(i, j) \leqslant k\}$ and its degree as $d_i \coloneq \#\{j: \dist(i, j) = 1\}$.

We sample a proportion $p$ of the nodes into the experiment, and randomly assign the nodes into two treatments: test and control. For node $i$, we denote whether it was sampled with $S_i \sim \Bern(p)$, and its assigned treatment as $D_i \sim \Bern(0.5)$. We will sometimes refer to nodes with $S_i = 1$, $D_i = 0$ as \emph{control} nodes and to nodes with $S_i = 1$, $D_i = 1$ as \emph{test} nodes.

Recall that the treatment is edge-level. For an edge between a sampled node and an unsampled node, we apply the treatment of the sampled node. For an edge between two sampled nodes with the same status, the edge inherits that shared treatment. However, for a \emph{conflicting} edge between a test node and a control node, we must decide (potentially in a randomized way) on a treatment assignment for this edge. Formally, the treatment of an edge $ij$ is given by
\begin{equation}
    Z_{ij} = Z_{ji} \coloneq \begin{cases}
        D_i & \text{if } S_i = 1, S_j = 0, \\
        D_j & \text{if } S_i = 0, S_j = 1, \\
        D_i = D_j & \text{if } S_i = S_j = 1, D_i = D_j, \\
        U_{ij} = U_{ji} & \text{if } S_i = S_j = 1, D_i \ne D_j, \\
        0 & \text{if } S_i = S_j = 0
    \end{cases}.
\label{eq:assignment_rule}
\end{equation}
For example, if we default to assigning control to a conflicting edge, then we are defining $U_{ij} \coloneq 0$. Alternatively, if we randomize, then we are defining $U_{ij} \sim \Bern(0.5)$. \Cref{fig:setup-examples} illustrates how the assignment rule resolves edge assignments around a focal control node.

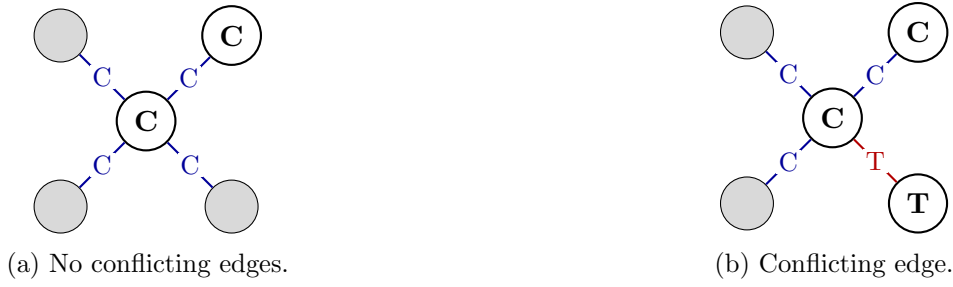
\begin{figure}[t]
    \centering
    \begin{subfigure}[b]{0.45\textwidth}
        \centering
        \begin{tikzpicture}[
            sampled/.style={circle, draw, thick, minimum size=7mm, font=\bfseries},
            unsampled/.style={circle, draw, fill=black!15, minimum size=7mm},
            cedge/.style={thick, blue!60!black},
            tedge/.style={thick, red!70!black},
            edgelabel/.style={fill=white, inner sep=1pt, font=\small},
        ]
            \node[sampled] (c) at (0, 0) {C};
            \node[sampled] (n1) at (45:1.6) {C};
            \node[unsampled] (n2) at (-45:1.6) {};
            \node[unsampled] (n3) at (-135:1.6) {};
            \node[unsampled] (n4) at (135:1.6) {};
            \draw[cedge] (c) -- node[edgelabel] {C} (n1);
            \draw[cedge] (c) -- node[edgelabel] {C} (n2);
            \draw[cedge] (c) -- node[edgelabel] {C} (n3);
            \draw[cedge] (c) -- node[edgelabel] {C} (n4);
        \end{tikzpicture}
        \caption{No conflicting edges.}
        \label{fig:setup-no-conflict}
    \end{subfigure}
    \hfill
    \begin{subfigure}[b]{0.45\textwidth}
        \centering
        \begin{tikzpicture}[
            sampled/.style={circle, draw, thick, minimum size=7mm, font=\bfseries},
            unsampled/.style={circle, draw, fill=black!15, minimum size=7mm},
            cedge/.style={thick, blue!60!black},
            tedge/.style={thick, red!70!black},
            edgelabel/.style={fill=white, inner sep=1pt, font=\small},
        ]
            \node[sampled] (c) at (0, 0) {C};
            \node[sampled] (n1) at (45:1.6) {C};
            \node[sampled] (n2) at (-45:1.6) {T};
            \node[unsampled] (n3) at (-135:1.6) {};
            \node[unsampled] (n4) at (135:1.6) {};
            \draw[cedge] (c) -- node[edgelabel] {C} (n1);
            \draw[tedge] (c) -- node[edgelabel] {T} (n2);
            \draw[cedge] (c) -- node[edgelabel] {C} (n3);
            \draw[cedge] (c) -- node[edgelabel] {C} (n4);
        \end{tikzpicture}
        \caption{Conflicting edge.}
        \label{fig:setup-conflict}
    \end{subfigure}
    \caption{Edge assignment rule and conflict resolution.}
    \label{fig:setup-examples}
    \vspace{1ex}
    \begin{minipage}{\textwidth}
        \footnotesize \textit{Note:} This figure presents examples of possible local configurations for a control node. Sampled nodes ($S_i = 1$) are shown as circles labeled by their treatment assignment $D_i \in \{C, T\}$. Unsampled nodes ($S_i = 0$) are shown as gray-filled circles without labels. Edge labels show the realized edge treatment $Z_{ij}$ from \Cref{eq:assignment_rule}. In (b), we have drawn the conflicting edge as assigned to $T$. Under equal edge weights $w_{ij} = 1/4$, the focal node's exposure is $H_i = 0$ in (a) and $H_i = 1/4$ in (b).
    \end{minipage}
\end{figure}

\begin{assumption}[Random assignment]
\label{asm:random}
All experiment assignments $S_i$, treatment assignments $D_i$, and exogenous randomness $U_{ij}$ are jointly independent, and $U_{ij}$ are all identically distributed.
\end{assumption}

We further assume that the experiment does not affect the graph, or, more generally, that the effect of the experiment on the graph is negligibly small.

\begin{assumption}[Exogenous graph]
\label{asm:exogenous-graph}
The graph $G$ and the weights $W$ are fixed.
\end{assumption}

We follow the Neyman--Rubin potential outcomes framework \citep[see, e.g.,][]{rubin2005causal} and treat the potential outcomes as fixed. We make three assumptions about how the treatment affects the outcome metric $Y_i$. First, we assume that the node's outcome only responds to treatments of incident edges.

\begin{assumption}[Interference is local]
    $Y_i$ depends on the sampling assignments, treatment assignments, and edge tie-breaks of the entire population only through the realized incident edge treatments $(Z_{ij})_j$.
\label{asm:locality}
\end{assumption}

Second, we assume the incident-edge treatments enter the outcome only through a single weighted summary, the exposure.

\begin{assumption}[Exposure aggregation]
    $Y_i$ depends on $(Z_{ij})_j$ only through the scalar exposure $H_i = \sum_j w_{ij} Z_{ij}$.
\label{asm:summary}
\end{assumption}

In practice, the true edge weights $w_{ij}$ are rarely known. We therefore consider the case in which the researcher constructs the exposure from another weighting $\tilde w_{ij}$ that need not coincide with the true $w_{ij}$ but shares its structure: it is nonnegative, supported on the edges of $G$, and normalized so that $\sum_j \tilde w_{ij} = 1$. A realistic example is the uniform weighting $\tilde w_{ij} = 1/d_i$, which uses only the adjacency structure of the graph. This gives the proxy exposure $\tilde H_i \coloneq \sum_j \tilde w_{ij} Z_{ij}$. Throughout, the true exposure $H_i$ is the special case $\tilde w = w$, so any statement we establish for the proxy exposure $\tilde H_i$ applies to $H_i$ as well.

Third, we assume the dependence of $Y_i$ on $H_i$ is linear.

\begin{assumption}[Linear dose-response]
    The outcome $Y_i$ is linear in $H_i$, namely $Y_i = a_i + b_i H_i$.
\label{asm:linearity}
\end{assumption}

With the constraint $\sum_j w_{ij} = 1$ for all nodes $i$, the exposure $H_i$ takes values in $[0, 1]$. For convenience, we also denote $\bar{a} = \bar{a}^{(n)} \coloneq \frac{1}{n} \sum_i a_i$ and $\bar{b} = \bar{b}^{(n)} \coloneq \frac{1}{n} \sum_i b_i$, using the $(n)$ superscript when we want to emphasize dependence on the population size. The estimand is GATE. When the control treatment is applied to the entire population, the average metric is $\frac{1}{n} \sum_i Y_i = \bar{a}$. Likewise, under the test treatment, the exposure $H_i$ for all nodes becomes $1$, and the average metric is $\frac{1}{n} \sum_i Y_i = \bar{a} + \bar{b}$. Therefore the estimand is given by
\[
    \thetan \coloneq \frac{1}{n} \sum_i b_i = \bar{b}.
\]

To avoid a single unit having out-sized influence on GATE, we also assume that the potential outcomes are bounded in \Cref{asm:bounded}.
\begin{assumption}[Bounded outcome]
    There exists a constant $M > 0$ such that $|a_i|, |a_i + b_i| < M$ for all nodes $i$.
\label{asm:bounded}
\end{assumption}

In online experiments, the population size is often extremely large. Formally, our results will be derived based on a sequence of networks indexed by the population size $n$, but suppressing the index $n$ where unambiguous. We further assume that the graph is sufficiently sparse in \Cref{asm:sparse-weak}.
\begin{assumption}[Weak sparsity]
    The maximum degree $\dmax \coloneq \max_i d_i = o(n^{1/2})$.
\label{asm:sparse-weak}
\end{assumption}

The role of \Cref{asm:sparse-weak} is to limit the dependence structure of the exposures $\{\tilde H_i\}$ (and hence the true exposures $\{H_i\}$ too, the special case $\tilde w = w$). Two exposures $\tilde H_i$ and $\tilde H_j$ are independent whenever $i$ and $j$ share neither an edge nor a common neighbor, since no underlying random variable then enters both sums $\tilde H_i = \sum_k \tilde w_{ik} Z_{ik}$ and $\tilde H_j = \sum_k \tilde w_{jk} Z_{jk}$. \Cref{asm:sparse-weak} bounds the number of such dependent pairs tightly enough for the law-of-large-numbers arguments in \Cref{sec:consistency}, which \Cref{lma:local-average} makes precise.

\section{Estimator}
\label{sec:estimator}

It might seem natural to simply regress the outcome $Y_i$ on the exposure $H_i$. However, ordinary least squares (OLS) is generally invalid in this network setting, even when the correct weights are known. Users may differ both in their responses to treatment and in their exposure distributions. In particular, given a node's own assignment, we expect the exposure of high-degree nodes to concentrate, whereas for low-degree nodes it stays more dispersed. The slope $b_i$ and the regressor $H_i$ are therefore not independent, which places us in a random coefficient model \citep{wooldridge1997two, wooldridge2003further}.
The next result quantifies this bias.

\begin{theorem}[Asymptotic bias of OLS]
    Under \Cref{asm:random,asm:exogenous-graph,asm:locality,asm:summary,asm:linearity,asm:bounded,asm:sparse-weak}, the OLS slope $\hat\beta_{\mathrm{OLS}}^{(n)}$ from regressing $Y_i$ on $H_i$ over the sampled units $\{i : S_i = 1\}$ satisfies
    \[
        \hat\beta_{\mathrm{OLS}}^{(n)} \toprob \sum_i b_i \cdot \frac{\var(H_i \mid S_i = 1)}{\sum_j \var(H_j \mid S_j = 1)}.
    \]
    It is consistent for GATE if and only if the empirical covariance between $b_i$ and $\var(H_i \mid S_i = 1)$ vanishes in the limit. Exposure variances $\var(H_i \mid S_i = 1)$ depend on the graph and are increasing in the weight concentration $\sum_j w_{ij}^2$.
\label{thm:ols}
\end{theorem}

We prove this in \Cref{app:proof-ols}. The OLS slope thus converges to a variance-weighted average of the node-level slopes $b_i$ rather than to GATE. This weighting may be undesirable, as it places lower weight on high-degree nodes, which are expected to carry higher business importance, as we observe in our empirical study (\Cref{sec:simulations}).

We address this with an instrumental variables strategy. For each node $i$, the assignment $D_i$ serves as an instrumental variable for its exposure: it is randomized and, by \Cref{asm:locality,asm:summary}, affects the outcome $Y_i$ only through the true exposure $H_i$. We therefore propose to estimate GATE using the Wald estimator built from the researcher's exposure $\tilde H_i$,
\begin{equation}
    \thetahat \coloneq \frac{\DeltaY^{(n)}}{\DeltaHt^{(n)}} = \frac{\frac{1}{N_1} \sum_i Y_i S_i D_i - \frac{1}{N_0} \sum_i Y_i S_i (1 - D_i)}{\frac{1}{N_1} \sum_i \tilde H_i S_i D_i - \frac{1}{N_0} \sum_i \tilde H_i S_i (1 - D_i)},
\label{eq:wald}
\end{equation}
where $N_1 = \sum_i S_i D_i$ and $N_0 = \sum_i S_i (1 - D_i)$ are the number of sampled units assigned to test and control, respectively. The numerator $\DeltaY^{(n)}$ and the denominator $\DeltaHt^{(n)}$ can be interpreted as differences in means for the outcome and exposure, respectively. Because $D_i$ is binary, the Wald estimator coincides with the two-stage least squares (2SLS) estimator that regresses $Y_i$ on $\tilde H_i$ using $D_i$ as an instrument, on the subsample of nodes with $S_i = 1$. In IV terminology, $\DeltaHt$ is the first-stage estimate, capturing how strongly the treatment assignment $D_i$ shifts the exposure $\tilde H_i$.

The denominator acts as a dilution factor. Since both control and test users have a mixed experience, the outcome difference $\DeltaY^{(n)}$ alone understates GATE. Dividing by the exposure difference rescales it to account for this dilution.
\section{Consistency}
\label{sec:consistency}

The following lemma is central to the consistency argument.

\begin{lemma}[Uniform compliance with instrumental variable]
    Under \Cref{asm:random,asm:exogenous-graph}, for each $d \in \{0, 1\}$ the conditional mean exposure $\E[H_i \mid S_i = 1, D_i = d] = \E[\tilde H_i \mid S_i = 1, D_i = d] = \zeta(d)$ is the same for all units $i$, for both the true exposure $H_i$ and any misspecified exposure $\tilde H_i$. Consequently the compliance $\zeta(1) - \zeta(0)$ is constant across $i$.
\label{lma:uniform-compliance}
\end{lemma}

\begin{proof}
    Fix any treatment value $d \in \{0, 1\}$ and let $\tilde H_i = \sum_j \tilde w_{ij} Z_{ij}$ be the exposure under any unit-sum weighting, so $\sum_j \tilde w_{ij} = 1$. The average exposure intensity is
    \begin{align*}
        \E[\tilde H_i \mid S_i = 1, D_i = d] = \E\left[\sum_j \tilde w_{ij} Z_{ij} \;\middle|\; S_i = 1, D_i = d\right] &= \sum_j \tilde w_{ij} \E[Z_{ij} \mid S_i = 1, D_i = d].
    \end{align*}
    Conditional on the node being sampled and on the treatment status $D_i = d$, the treatment for edge $ij$ only depends on $S_j$, $D_j$, and $U_{ij}$. By \Cref{asm:random}, the value $\E[Z_{ij} \mid S_i = 1, D_i = d]$ does not depend on $(i, j)$, and we denote this common value by $\zeta(d)$. Therefore,
    \begin{align*}
        \E[\tilde H_i \mid S_i = 1, D_i = d] = \sum_j \tilde w_{ij}\, \zeta(d) = \zeta(d),
    \end{align*}
    where the second equality uses $\sum_j \tilde w_{ij} = 1$. The compliance is then the difference of two such constants and is also constant across $i$.
\end{proof}

For the particular assignment rule specified in \Cref{eq:assignment_rule}, the average exposure equals $\zeta(d) = (1-p/2) d + pu/2$, where $u = \E[U_{ij}]$ is the expected treatment assigned to a conflicting edge. The compliance then equals $\zeta(1) - \zeta(0) = 1 - p/2$, independent of $u$.

\begin{theorem}[Consistency]
    Under \Cref{asm:random,asm:exogenous-graph,asm:locality,asm:summary,asm:linearity,asm:bounded,asm:sparse-weak}, and for any proxy weighting $\tilde w$, the estimator $\hat\theta$ is consistent for $\theta$, i.e., $\thetahat - \thetan \toprob 0$ as $n \to \infty$. Along the way, we also have $\DeltaHt^{(n)} - (\zeta(1) - \zeta(0)) \toprob 0$.
\label{thm:consistency}
\end{theorem}

The proof is deferred to \Cref{app:proof-consistency}. There we show that $\thetahat$ converges in probability to
\begin{align*}
    &~ \frac{\frac{1}{n} \sum_i \E[Y_i \mid S_i = 1, D_i = 1] - \frac{1}{n} \sum_i \E[Y_i \mid S_i = 1, D_i = 0]}{\frac{1}{n} \sum_i \E[\tilde H_i \mid S_i = 1, D_i = 1] - \frac{1}{n} \sum_i \E[\tilde H_i \mid S_i = 1, D_i = 0]} \\
    =&~ \frac{\frac{1}{n} \sum_i (a_i + b_i \E[H_i \mid S_i = 1, D_i = 1]) - \frac{1}{n} \sum_i (a_i + b_i \E[H_i \mid S_i = 1, D_i = 0])}{\frac{1}{n} \sum_i \E[\tilde H_i \mid S_i = 1, D_i = 1] - \frac{1}{n} \sum_i \E[\tilde H_i \mid S_i = 1, D_i = 0]} \\
    =&~ \sum_i b_i \cdot \frac{\E[H_i \mid S_i = 1, D_i = 1] - \E[H_i \mid S_i = 1, D_i = 0]}{\sum_j (\E[\tilde H_j \mid S_j = 1, D_j = 1] - \E[\tilde H_j \mid S_j = 1, D_j = 0])} \\
    =&~ \sum_i b_i \cdot \frac{\zeta(1) - \zeta(0)}{\sum_j (\zeta(1) - \zeta(0))} = \bar{b} = \thetan.
\end{align*}
In the correctly specified case $\tilde w = w$, the third line expresses the limit as a weighted average of $b_i$ with weights proportional to each unit's compliance $\E[H_i \mid S_i = 1, D_i = 1] - \E[H_i \mid S_i = 1, D_i = 0]$. This is consistent with the classical IV literature on non-binary endogenous treatments. \citet{angrist1995twostage} shows that, in general, the IV estimator converges to a weighted average of marginal causal responses, with weights proportional to the responsiveness of the treatment to the instrument. \citet{angrist2000interpretation} extends this to a continuous endogenous treatment, where the limit also averages each unit's marginal effect over the treatment values its instrument-induced shift covers. Linearity (\Cref{asm:linearity}) collapses this within-unit averaging by making the marginal effect a single slope $b_i$.

\Cref{lma:uniform-compliance} shuts down the remaining source of heterogeneity by forcing the compliance differences to a common value, which makes the weights uniform and collapses the weighted average in the last equality to a simple sample mean. The uniformity of compliance is a structural consequence of the assignment rule rather than an assumption on behavior: the mapping from the instrument $D_i$ to the exposure $H_i$ is determined by the platform, not by endogenous take-up decisions of users. This is what allows us to derive a stronger identification result than in the classical IV literature.

This is also why the choice of weights does not affect consistency. Our design assigns the same expected treatment to every edge of a node, $\E[Z_{ij} \mid S_i = 1, D_i = d] = \zeta(d)$, independent of the edge $ij$. Any unit-sum weighting averages these identical edge means back to $\zeta(d)$, so the first stage $\DeltaHt^{(n)}$ has the same limit regardless of $\tilde w$. While we could use the closed-form compliance $\zeta(1) - \zeta(0) = 1 - p/2$ directly in the estimator, it can be harder to derive in more complex settings relevant to online platforms, as in \Cref{sec:holdout,sec:multi-arm,sec:experiment-conflicts}. We therefore estimate the compliance from the data through the first stage $\DeltaHt^{(n)}$. We discuss this choice in \Cref{sec:denominator}.

Finally, we characterize what the Wald estimator targets if the dose-response relationship is nonlinear. Suppose that $Y_i=y_i(H_i)$ for a fixed response function $y_i:[0,1]\to\mathbb{R}$. For a sampled node, define $H_i(d)$ by setting $D_i=d$ in the assignment rule while leaving all other assignments and edge tie-breaks random. The next result adapts the IV weighting results of \citet{angrist2000interpretation} to our design.

\begin{theorem}[Wald limit under nonlinear dose-response]
    Suppose the response functions $y_i$ are uniformly bounded and absolutely continuous, with uniformly bounded derivatives. Under \Cref{asm:random,asm:exogenous-graph,asm:locality,asm:summary,asm:sparse-weak}, for any proxy weighting $\tilde w$, the Wald estimator converges to the average causal response (ACR), defined as
    \begin{equation}
        \thetahat-\theta_{\mathrm{ACR}}^{(n)}\toprob0,
        \qquad
        \theta_{\mathrm{ACR}}^{(n)}
        \coloneq
        \frac{1}{n}\sum_i\int_0^1 y_i'(h)\omega_i(h)\,dh,
    \label{eq:nonlinear-acr}
    \end{equation}
    where
    \begin{equation}
        \omega_i(h)
        \coloneq
        \frac{
            \PP\big(H_i(0)<h\leqslant H_i(1)\mid S_i=1\big)
        }{\zeta(1)-\zeta(0)}
        \geqslant0,
        \qquad
        \int_0^1\omega_i(h)\,dh=1.
    \label{eq:nonlinear-weight}
    \end{equation}
\label{thm:nonlinear-weighting}
\end{theorem}

The proof is given in \Cref{app:proof-nonlinear-weighting}. Every user still receives the same overall weight, while the concentration of its edge weights determines which part of its response curve is emphasized. The weights $\omega_i(h)$ measure how often changing a node's assignment moves its exposure across each level $h$. \Cref{fig:nonlinear-weighting} illustrates these weights. For example, for users with many similarly weighted connections, $H_i(0)$ and $H_i(1)$ concentrate around $p/4$ and $1-p/4$, respectively, so exposure levels near zero or one receive little weight. Consequently, with nonlinear GATE defined as $\theta_{\mathrm{GATE}}^{(n)}\coloneq n^{-1}\sum_i\big(y_i(1)-y_i(0)\big)$, Wald overstates GATE when responses are stronger in the central exposure range than near zero and one, and understates GATE when responses are stronger in the tails.

\section{Variance estimation and asymptotic normality}
\label{sec:normality}

In this section, we establish the limiting distribution of the error $\thetahat - \thetan$ as the population size $n \to \infty$.

\subsection{Asymptotic normality}
\label{sec:asymptotic-normality}

To analyze the asymptotic distribution of this ratio, we linearize the estimation error $\thetahat - \thetan$. We can rewrite the error as
\begin{equation}
    \thetahat - \thetan = \frac{\DeltaY^{(n)} - \thetan \DeltaHt^{(n)}}{\DeltaHt^{(n)}}.
\label{eq:linear}
\end{equation}

From \Cref{thm:consistency}, we have $\DeltaHt^{(n)} - (\zeta(1) - \zeta(0)) \toprob 0$. So by Slutsky's theorem, it suffices to determine the asymptotic distribution of the numerator. Because the numerator is centered, the fluctuations of the denominator around its limit do not affect the first-order variance.

\begin{lemma}
    Under \Cref{asm:random,asm:exogenous-graph,asm:locality,asm:summary,asm:linearity,asm:bounded,asm:sparse-weak}, $\DeltaY^{(n)} - \thetan \DeltaHt^{(n)}$ is asymptotically linear. Specifically,
    \[
        \DeltaY^{(n)} - \thetan \DeltaHt^{(n)} = \frac{1}{n} \sum_i \psi_i + o_p(n^{-1/2}),
    \]
    where $\psi_i = \frac{2}{p} S_i (2D_i - 1)(Y_i - \thetan \tilde H_i - \abarn)$.
\label{lma:linearize}
\end{lemma}
We prove this in \Cref{app:proof-linearize}. We define $\sigman^2$ as the asymptotic variance of the linearized numerator $\DeltaY^{(n)} - \thetan \DeltaHt^{(n)}$,
\[
    \sigman^2 \coloneq \var\left(\frac{1}{\sqrt{n}} \sum_i \psi_i\right).
\]
In the particular case of no treatment effects, $b_i = 0$ for all $i$ and $\thetan = 0$, so that $Y_i - \thetan \tilde H_i - \abarn$ reduces to $a_i - \abarn$. The influence terms are then independent across $i$, with $\var(\psi_i) = \frac{4}{p} (a_i - \abarn)^2$, and
\[
    \sigman^2 = \frac{4}{p} \cdot \frac{1}{n} \sum_i (a_i - \abarn)^2,
\]
which coincides with the asymptotic variance of the standard difference-in-means estimator on the test and control subsamples. Without treatment effects, all of the variance of the numerator comes from the dispersion of the baseline outcomes and does not depend on the graph.

For the linearization argument to apply, we also need to avoid superefficiency of the estimator. This excludes degenerate cases such as $Y_i = c$ for all $i$, in which case the estimator $\thetahat$ equals zero exactly, with a distribution that is a point mass at zero. We rule this out via the following assumption.

\begin{assumption}[No superefficiency]
    The asymptotic variance is bounded away from zero, that is, $\liminf_n \sigman^2 > 0$.
\label{asm:no-superefficient}
\end{assumption}

For the central limit argument, we also strengthen the graph sparsity condition in \Cref{asm:sparse-weak} as follows.

\begin{assumption}[Medium sparsity]
    The maximum degree is $\dmax = o(n^{1/4})$.
\label{asm:sparse-medium}
\end{assumption}

\begin{theorem}[Asymptotic normality]
    Under \Cref{asm:random,asm:exogenous-graph,asm:locality,asm:summary,asm:linearity,asm:bounded,asm:sparse-medium,asm:no-superefficient}, the estimation error converges in distribution
    \[
        \frac{\sqrt{n}(\thetahat - \thetan)}{\sigman / (\zeta(1) - \zeta(0))} \todist \Normal(0, 1).
    \]
\label{thm:normal}
\end{theorem}

The proof, which uses Stein's method for local dependence, is given in \Cref{app:proof-normal}. The variance $\sigman^2$ depends on the weighting $\tilde w$, unlike the probability limit $\thetan$. Misspecifying the weights therefore leaves the point estimate intact but can change the variance. In \Cref{app:proof-efficiency} we show that under homogeneous effects the true weights $\tilde w = w$ minimize the variance. Empirically, however, we find that the choice of weights has little effect on the variance (\Cref{sec:simulations}).

\subsection{Conservative variance estimator}
\label{sec:conservative-variance}

To conduct inference, we first construct a conservative estimator for $\sigman^2$. Because $\psi_i$ and $\psi_j$ are independent when $\dist(i, j) > 2$, we begin with a sample analogue that aggregates products within the dependency graph. Denoting $N \coloneq N_1 + N_0$, we estimate the unobserved mean baseline outcome $\bar{a}$ with
\[
    \hat{a} = \frac{1}{N}\sum_i S_i (Y_i - \hat\theta \tilde H_i).
\]
We can then define the empirical unscaled error $\hat{X}_i = Y_i - \hat\theta \tilde H_i - \hat{a}$, and the corresponding empirical linear terms
\[
    \hat\psi_i = \frac{2}{p}S_i(2D_i - 1)\hat{X}_i.
\]
The proposed variance estimator is the following quadratic form over pairs of nodes within two hops:
\begin{equation}
    \sigmanhat^2 = \frac{1}{n} \sum_i \sum_{j \in N_2(i)} \hat\psi_i \hat\psi_j.
\label{eq:variance-estimator}
\end{equation}
Because $S_i$ indicates whether node $i$ is sampled, $\hat\psi_i=0$ for every unsampled node, so only pairs of sampled nodes contribute to $\sigmanhat^2$. Computing the estimator nevertheless requires enough graph information to identify sampled nodes within two hops, including paths through unsampled intermediate nodes.

Bounding the sampling error of $\sigmanhat^2$ requires a stricter sparsity condition than \Cref{asm:sparse-medium}.

\begin{assumption}[Strong sparsity]
    The maximum degree is $\dmax = o(n^{1/6})$.
\label{asm:sparse-strong}
\end{assumption}

\begin{theorem}[Probability limit of the variance estimator]
    Under \Cref{asm:random,asm:exogenous-graph,asm:locality,asm:summary,asm:linearity,asm:bounded,asm:sparse-strong},
    \[
        \sigmanhat^2 - \sigman^2 - \Gamma_n \toprob 0,
    \]
    where
    \[
        \Gamma_n \coloneq \left(\zeta(1) - \zeta(0)\right)^2
        \frac{1}{n} \sum_i \sum_{j \in N_2(i)}
        (b_i - \bar b)(b_j - \bar b).
    \]
\label{thm:variance-limit}
\end{theorem}

See \Cref{app:proof-variance-limit} for the proof. The discrepancy $\Gamma_n$ arises because the individual influence terms need not have mean zero. The asymptotic variance $\sigman^2$ aggregates local covariances, whereas the variance estimator $\sigmanhat^2$ aggregates local raw cross-products. Their difference is the product of the influence-term means, which yields $\Gamma_n$. We impose the following restriction on the spatial correlation of heterogeneous treatment effects.

\begin{assumption}[Nonnegative local effect alignment]
    The unnormalized local covariance of treatment effects, computed within two-hop neighborhoods, is asymptotically nonnegative:
    \[
        \liminf_{n \to \infty} \Gamma_n \geqslant 0.
    \]
\label{asm:nonnegative-local-alignment}
\end{assumption}

This assumption is plausible when nearby nodes tend to have similar treatment effects, so local cross-products contribute positively, or when centered treatment effects of distinct nearby nodes are locally uncorrelated, so their cross-products average to zero. It permits some negative pairwise alignment because $\Gamma_n$ also includes the nonnegative terms $(b_i-\bar b)^2$. It holds automatically under homogeneous treatment effects.\footnote{The assumption can fail when treatment effects are strongly negatively correlated within two-hop neighborhoods. For example, consider a cycle graph with the repeating pattern $b_i-\bar b=\delta(1,1,-1,-1,\ldots)$. For each node, the self-product is $\delta^2$, the two products with adjacent nodes sum to zero on average, and the two products with nodes at distance two sum to $-2\delta^2$. Thus the average product over the two-hop neighborhood is $-\delta^2$, giving $\Gamma_n=-\left(\zeta(1)-\zeta(0)\right)^2\delta^2$ and making the variance estimator asymptotically anti-conservative.}

\begin{corollary}[Conservative confidence intervals]
    Under \Cref{asm:random,asm:exogenous-graph,asm:locality,asm:summary,asm:linearity,asm:bounded,asm:sparse-strong,asm:no-superefficient,asm:nonnegative-local-alignment}, the Wald confidence interval has asymptotic coverage at least $1 - \alpha$:
    \[
        \liminf_{n \to \infty} \PP\left(
            \thetan \in \left[
                \thetahat \pm z_{1 - \alpha/2}
                \frac{\sigmanhat}{\sqrt{n}\,\DeltaHt^{(n)}}
            \right]
        \right) \geqslant 1 - \alpha.
    \]
\label{cor:conservative-ci}
\end{corollary}

See \Cref{app:proof-conservative-ci} for the proof.

\subsection{Approximate variance estimator}
\label{sec:diagonal-variance}

Computing \Cref{eq:variance-estimator} requires the adjacency matrix and the two-hop neighborhoods. These can be computationally costly to work with on a large graph, and they may not be available at all when graph information is not retained after the experiment. Dropping the one-hop and two-hop blocks gives the diagonal approximation
\[
    \sigmandiag^2 \coloneq \frac{1}{n} \sum_i \hat\psi_i^2 = \frac{4}{p^2} \cdot \frac{1}{n} \sum_i S_i \big(Y_i - \hat\theta \tilde H_i - \hat{a}\big)^2,
\]
which ignores the network correlation of the influence terms $\hat\psi_i$. The following result quantifies what is lost.

\begin{theorem}[Second-order bias of the diagonal variance estimator]
    Under \Cref{asm:random,asm:exogenous-graph,asm:locality,asm:summary,asm:linearity,asm:bounded,asm:sparse-strong}, with correct weights $\tilde w = w$ and a symmetric tie-break $u = \E[U_{ij}] = 1/2$, the bias of $\sigmandiag^2$ is asymptotically a quadratic form in the centered slopes,
    \[
        \sigmandiag^2 - \sigman^2 \toprob \frac{1}{n} \sum_i \sum_{j \in N_2(i)} (b_i - \bar b)(b_j - \bar b) \Lambda_{ij},
    \]
    where the coefficients $\Lambda_{ij}$ are bounded.
\label{thm:diagonal-quadratic}
\end{theorem}

See \Cref{app:proof-diagonal-quadratic} for the proof. The bias is second order in the dispersion of the treatment effects. It is therefore small relative to $\sigman^2$ whenever the slopes $b_i$ vary little relative to the dispersion of the baseline outcomes $a_i$ and the average two-hop neighborhood size remains bounded. We show in \Cref{sec:simulations} that this approximation has good properties in our empirical simulations.

\section{Empirical study}
\label{sec:simulations}

\begin{table}[t]
    \centering
    \caption{Mean estimates and standard errors by estimator and outcome metric.}
    \label{tab:simulation-mean-estimates}
    \begin{tabular}{llcccc}
        \toprule
        & & \multicolumn{2}{c}{Mean estimate} & \multicolumn{2}{c}{Standard error} \\
        \cmidrule(lr){3-4} \cmidrule(lr){5-6}
        Estimator & Exposure & Total activity & Days active & Total activity & Days active \\
        \midrule
        Wald & Correct & 0.999 & 0.999 & 0.0593 & 0.0119 \\
        Wald & Misspecified & 0.999 & 0.999 & 0.0593 & 0.0119 \\
        Diff-in-means & --- & 0.749 & 0.749 & 0.0445 & 0.0089 \\
        OLS & Correct & 0.964 & 0.933 & 0.0595 & 0.0109 \\
        OLS & Misspecified & 0.889 & 0.908 & 0.0534 & 0.0109 \\
        ERL & Correct & 1.001 & 1.000 & 0.0567 & 0.0112 \\
        ERL & Misspecified & 1.001 & 1.000 & 0.0574 & 0.0113 \\
        \bottomrule
    \end{tabular}
    \par\vspace{2ex}
    \begin{minipage}{0.8\textwidth}
        \footnotesize \textit{Note:} The table reports relative treatment-effect estimates across 20,000 Monte Carlo draws. All entries are reported in percentage points. Standard errors are the standard deviations of the estimates across simulation draws. The experiment sampling proportion is 50\%, and the simulated relative treatment effect is 1\%.
    \end{minipage}
\end{table}

We evaluate the performance of our estimator using data from an online communication platform. Our simulations reflect a hypothetical experiment improving the communication channel. We use one week of data and construct an undirected network from users engaging in pairwise interactions. The nodes are users, and two users are joined by an edge whenever they had at least one observed interaction. The weight $w_{ij}$ is the share of user $i$'s interactions with user $j$, which is asymmetric and sums to one over each node's edges. In each Monte Carlo draw, we sample a proportion $p$ of users, randomly assign sampled users to test or control, and resolve conflicting endpoint assignments according to \Cref{eq:assignment_rule}.

We consider two user-level outcome metrics, total interaction activity and number of days active. Total activity is heavy-tailed and has higher variance than days active, resulting in larger standard errors in our simulations. We take each user's observed value of the metric as the baseline $a_i$ and generate outcomes multiplicatively, $Y_i = a_i (1 + \tau H_i)$, where $H_i$ is the true interaction-weighted exposure and $\tau$ is the relative treatment effect. The slope $b_i = \tau a_i$ therefore varies across users, so the design carries heterogeneous treatment effects. The estimand in relative units, $\thetan / \bar{a}$, is exactly $\tau$, and we report all estimates on that scale because the absolute levels of the metrics are proprietary. The misspecified exposure $\tilde H_i = \sum_j \tilde w_{ij} Z_{ij}$ instead weights every user pair equally, $\tilde w_{ij} = 1 / d_i$.

\Cref{tab:simulation-mean-estimates} presents mean estimates for the two outcome metrics under our baseline specification, with a 1\% relative treatment effect and a 50\% experiment sampling proportion. The Wald estimates remain close to the true effect with either the correct or misspecified exposure weights, reflecting \Cref{thm:consistency}, which establishes the consistency of the Wald estimator even when the exposure weights are misspecified. The difference-in-means estimator is strongly downward biased because treatment dilution shrinks the contrast in exposure between test and control. OLS is downward biased even with the correct exposure weights, and the bias is larger when the weights are misspecified. This result reflects \Cref{thm:ols}, which shows that OLS is generally biased under heterogeneous treatment effects. In our data, users with higher engagement metric values have more interaction partners and therefore lower weight concentration $\sum_j w_{ij}^2$, so OLS weights them least while their slopes $b_i = \tau a_i$ are largest, pulling the estimate below the true effect.

We also test the performance of ERL under the same estimation design.\footnote{ERL of \citet{harshaw2023design} estimates GATE by $\hat\theta_{\mathrm{ERL}}^{(n)} \coloneq \frac{1}{N} \sum_i S_i (H_i - \mu_i) Y_i / v_i$, where $\mu_i \coloneq \E[H_i \mid S_i = 1]$ and $v_i \coloneq \var(H_i \mid S_i = 1)$. In the simulation we compute $\mu_i$ and $v_i$ per unit in closed form, using $\mu_i = 1/2$ from \Cref{lma:uniform-compliance} and $v_i$ from \Cref{eq:vi-decomp}. We apply centering that replaces $Y_i$ by $Y_i - \bar a$, which leaves the estimand unchanged while removing the estimator's dependence on the realized mean of the scores and decreasing its variance.} ERL is unbiased under both the correct and the misspecified weights.\footnote{This is not a general property of ERL, but a consequence of our randomization design and the uniform weighting $\tilde w_{ij} = 1/d_i$, for which $\cov(H_i, \tilde H_i \mid S_i = 1)$ and $\var(\tilde H_i \mid S_i = 1)$ coincide. We omit the details.} Its standard errors are slightly smaller than those of the Wald estimator. However, as we discuss in \Cref{sec:erl-comparison}, the small variance of ERL is a consequence of our randomization design and of the default setting, in which no unit has near-zero exposure variance to inflate its inverse-variance weight. This property is hard to maintain in practice, where we expect the Wald estimator to have an advantage in statistical power.

\begin{figure}[t]
    \centering
    \includegraphics[width=0.75\textwidth]{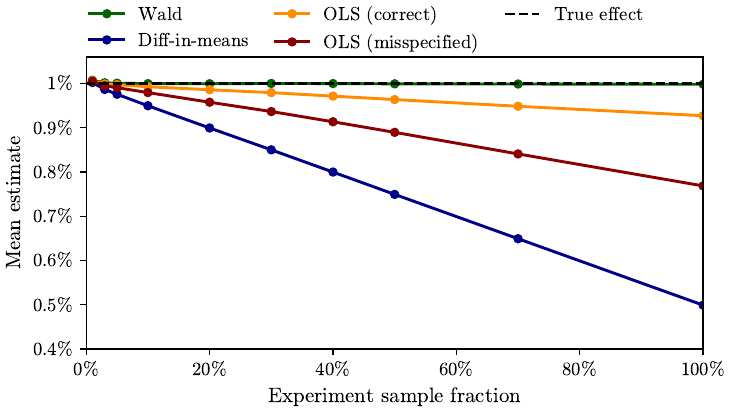}
    \caption{Mean treatment-effect estimates.}
    \label{fig:simulation-sample-fraction}
    \vspace{1ex}
    \begin{minipage}{\textwidth}
        \footnotesize \textit{Note:} The figure reports mean estimates for total activity across Monte Carlo simulations as the experiment sampling proportion varies, holding the relative treatment effect fixed at 1\%. Each point uses 20,000 Monte Carlo draws at sampling proportions of 5\%, 20\%, and 50\%, and 5,000 draws at the remaining proportions. The displayed Wald curve uses the correct exposure weights (the Wald curves under the correct and misspecified weights coincide).
    \end{minipage}
\end{figure}

\begin{figure}[t]
    \centering
    \includegraphics[width=\textwidth]{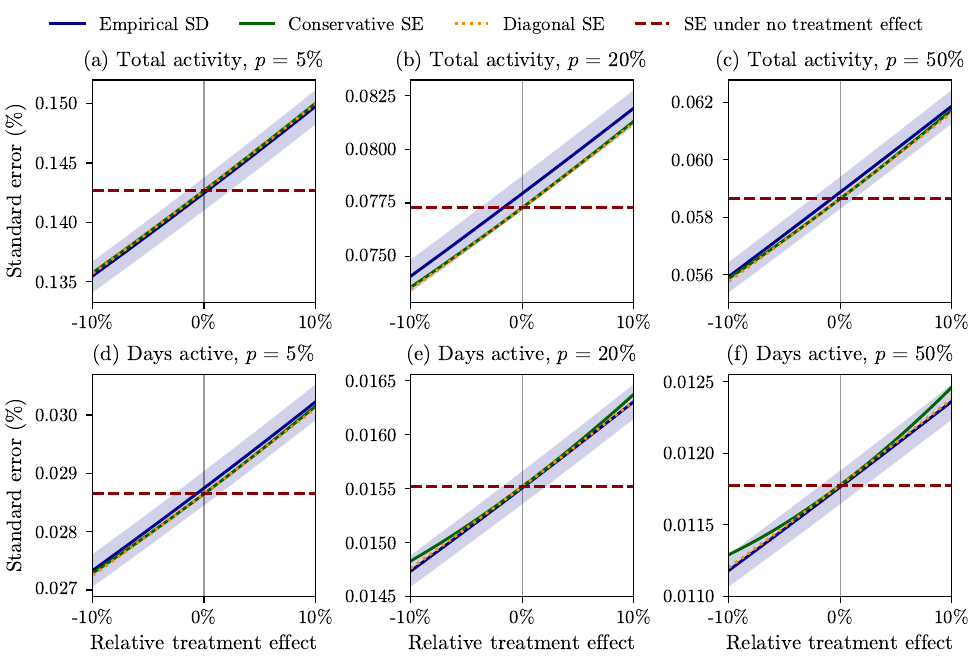}
    \caption{Standard error of the Wald estimator.}
    \label{fig:simulation-standard-error}
    \vspace{1ex}
    \begin{minipage}{\textwidth}
        \footnotesize \textit{Note:} Each panel varies the relative treatment effect at a fixed sampling proportion, with total activity on the top row and days active on the bottom. The conservative and diagonal curves are the implied standard errors of $\sigmanhat^2$ and $\sigmandiag^2$, taken as medians across 100 independent simulation draws, and the third curve is the implied standard error of $\sigman^2$ under no treatment effect. The empirical standard deviation is the realized dispersion of the estimates across 20,000 Monte Carlo draws, and the shaded band around it is a 95\% bootstrap interval obtained by resampling those draws. All results are for correct weights.
    \end{minipage}
\end{figure}

\Cref{fig:simulation-sample-fraction} illustrates the effect of treatment dilution as the experiment sampling proportion increases.
With a larger sample, more edges connect nodes with conflicting assignments, making the attenuation of the difference-in-means estimator substantial. By contrast, the Wald estimator corrects for the diluted exposure contrast and remains close to the true effect across sampling proportions. The OLS bias also gets worse as the experiment sampling proportion increases.

We next study the asymptotic distribution of the Wald estimator. The histograms in \Cref{fig:simulation-distributions} show that the sampling distributions of the Wald estimator are approximately normal for both outcome metrics, in line with the asymptotic normality established in \Cref{thm:normal}. Anderson--Darling tests do not reject normality for either metric: the test statistic is $A^2 = 0.41$ for total activity and $A^2 = 0.32$ for days active, both below the 5\% critical value of $0.787$.

Finally, we compare the properties of different variance estimators. \Cref{fig:simulation-standard-error} presents results for correct weights. We take the empirical standard deviation of $\thetahat$ across simulation draws as the benchmark, since it is the realized sampling dispersion that $\sigman^2$ is meant to approximate. We show the results for variance estimators based on \Cref{sec:conservative-variance} and \Cref{sec:diagonal-variance}. The conservative and the diagonal standard errors both provide a good approximation to the benchmark. Over the range of treatment effects and sampling proportions we consider, they lie within the bootstrap interval of the empirical standard deviation at every point of the grid. We also present the variance under no treatment effect derived in \Cref{sec:asymptotic-normality}. For small treatment effects both the conservative and the diagonal standard errors remain close to it, while for larger effects they provide a better approximation. We find that under the misspecified weights the performance of both variance estimators stays the same.

We also find that \Cref{asm:nonnegative-local-alignment} is satisfied for both metrics, so the estimator in \Cref{eq:variance-estimator} is conservative. Its upward bias is small, however, because $\Gamma_n$, which is quadratic in the treatment effects, is small relative to the variance contributed by baseline outcomes. We observe higher conservativeness for days active. This reflects a property of the metric: relative to its own dispersion, days active is more strongly aligned across neighboring users than total activity.

\section{Practical considerations}
\label{sec:practical}

This section briefly discusses several extensions that could matter in deployed experiments.

\subsection{Covariate adjustment}

Our method can be combined with standard variance reduction techniques such as covariate adjustment. One way to do it while staying within the Wald framework is to use out-of-experiment users to regress the outcome $Y_i$ on exogenous covariates $X_i$, obtaining a fitted function $\hat\mu(X_i)$, and then apply the Wald estimator to the adjusted outcomes $\tilde Y_i = Y_i - \hat\mu(X_i)$ rather than the raw $Y_i$. This effectively replaces the baseline in the linear model $Y_i = a_i + b_i H_i$ with $\tilde a_i = a_i - \hat\mu(X_i)$, leaving the slopes $b_i$ and the estimand GATE unchanged while reducing the baseline variability. This is illustrated by the no-treatment-effect result in \Cref{sec:asymptotic-normality}: the asymptotic variance scales with the population variance of the baselines, so shrinking the baseline residuals through covariate adjustment directly reduces it.

\subsection{Holdout group}
\label{sec:holdout}

In the context of online experimentation, there is often a long-running holdout group of users to whom new features are never rolled out. Our method can be extended to accommodate this practical constraint. Formally, we assume there exists a Bernoulli random variable $F_i \sim \Bern(f)$ for each node $i$, with $F_i = 1$ indicating that node $i$ belongs to the frozen holdout group. We assume holdout is selected randomly, so the $F_i$ are jointly independent of $S_i, D_i, U_{ij}$ and satisfy the analogue of \Cref{asm:random}. The rule described in \Cref{eq:assignment_rule} can be extended so that $Z_{ij} = 0$ whenever $F_i = 1$ or $F_j = 1$, ensuring that no edge incident to a frozen node receives the treatment. It is easy to see that \Cref{lma:uniform-compliance} extends to this setting as well, with the compliance
\[
    \E[\tilde H_i \mid S_i = 1, F_i = 0, D_i = 1] - \E[\tilde H_i \mid S_i = 1, F_i = 0, D_i = 0]
\]
constant across all non-frozen units $i$. Therefore, $D_i$ can still be used as an instrument for the exposure of sampled, non-frozen users, and the Wald estimator delivers consistent estimates of GATE. Concretely, this compliance equals $(1-f)(1-p/2)$. When the holdout share $f$ is substantial, the IV correction is needed even for small experiments (small $p$), since conflicts with frozen neighbors generate substantial dilution of user experience.

\subsection{Multi-arm experiments}
\label{sec:multi-arm}

The Wald estimator can be applied to experiments with several arms. Suppose each sampled node $i$ is assigned to one of $K + 1$ arms $D_i \in \{0, 1, \ldots, K\}$, where $D_i = 0$ denotes control. The assignment rule from \Cref{eq:assignment_rule} extends to this case by resolving conflicts between different arms uniformly at random, similarly to how conflicts between test and control were resolved in the binary case.\footnote{In the multi-arm setting, $U_{ij}$ must be strictly $\Bern(0.5)$ so that conflicts between control and a test arm and conflicts between two test arms both resolve with the same probability $1/2$. Otherwise the asymmetry between the two conflict types breaks the exclusion restriction below.} Let $H_i^{(k)} \coloneq \sum_j w_{ij} \mathds{1}\{Z_{ij} = k\}$ denote node $i$'s exposure to arm $k$, with the proxy exposure $\tilde H_i^{(k)}$ defined analogously. We can then estimate the effect of arm $k$ relative to control as
\[
    \hat\theta_k = \frac{\frac{1}{N_k} \sum_i Y_i S_i \mathds{1}\{D_i = k\} - \frac{1}{N_0} \sum_i Y_i S_i \mathds{1}\{D_i = 0\}}{\frac{1}{N_k} \sum_i \tilde H_i^{(k)} S_i \mathds{1}\{D_i = k\} - \frac{1}{N_0} \sum_i \tilde H_i^{(k)} S_i \mathds{1}\{D_i = 0\}},
\]
where $N_k = \sum_i S_i \mathds{1}\{D_i = k\}$ is the number of sampled nodes assigned to arm $k$. It is easy to see that \Cref{lma:uniform-compliance} extends to this setting as well, with the compliance for arm $k$
\[
    \E[\tilde H_i^{(k)} \mid S_i = 1, D_i = k] - \E[\tilde H_i^{(k)} \mid S_i = 1, D_i = 0],
\]
constant across all sampled units $i$. Moreover, for any other arm $\ell \ne k$, the expected exposure is the same under $D_i = k$ and $D_i = 0$,
\[
    \E[H_i^{(\ell)} \mid S_i = 1, D_i = k] = \E[H_i^{(\ell)} \mid S_i = 1, D_i = 0],
\]
since $Z_{ij} = \ell$ requires a sampled neighbor with $D_j = \ell$ to win a tie-break against $D_i$, and the probability of this does not depend on the specific value of $D_i$ as long as $D_i \ne \ell$. The exclusion restriction therefore holds. However, to extend the theoretical results of \Cref{sec:consistency} on consistency to the multi-arm case, we would need a stronger linearity assumption that replaces \Cref{asm:linearity} with linearity with respect to exposures of all experimental arms,
\begin{equation}
    Y_i = a_i + \sum_{k=1}^K b_i^{(k)} H_i^{(k)}.
\label{eq:multi-arm-linearity}
\end{equation}
In addition to this, multi-arm experiments rely on linearity more heavily whenever the per-arm sample fraction is held fixed to preserve estimation power. The total sample $p$ then grows linearly with $K$, increasing the share of edges in conflict and pushing each arm's exposure further from $0$ and $1$. By \Cref{thm:nonlinear-weighting}, extrapolation to the full-rollout counterfactual is then more sensitive to deviations from linearity because the estimator places less weight on marginal responses near the endpoints.

\subsection{Experiments with conflicts}
\label{sec:experiment-conflicts}

Online platforms sometimes concurrently run several independent experiments that affect the same surface or product parameter. In this case, two treatments cannot be assigned to the same edge at the same time. Our framework can be applied to this setting as well, with disjoint node samples across experiments. Conflicts on edges connecting nodes in different experiments are resolved at random. Similarly, this extension requires outcomes to be linear in exposure to every experiment's treatment, as in \Cref{eq:multi-arm-linearity}. The Wald estimator scales the outcome difference by the corresponding exposure difference within each experiment. The validity argument is the same as in \Cref{sec:multi-arm}.

\subsection{Non-monotonicity of power}

An interesting feature of our design is that statistical power may be non-monotonic in the sampling proportion $p$. Increasing $p$ adds more sampled units, which typically reduces variance. However, a larger $p$ also increases the fraction of conflicting edges (edges between sampled nodes with different treatment assignments), which dilutes the experience of both test and control nodes, shrinking the difference in their average exposure. Since the asymptotic variance of $\thetahat$ scales as $\sigman^2 / (1 - p/2)^2$, the denominator shrinks as $p$ increases, potentially offsetting the variance reduction from a larger sample. For example, under no treatment effect the variance is proportional to $1 / [p (1 - p/2)^2]$, which is minimized at $p^* = 2/3$. Moreover, even for $p < p^*$ the returns from increasing $p$ are smaller than the standard $\frac{1}{\sqrt p}$ scaling would suggest, since the dilution factor $(1-p/2)^2$ erodes part of the gain.

\subsection{Choice of denominator}
\label{sec:denominator}

In the Wald estimator \Cref{eq:wald}, the denominator $\DeltaHt^{(n)}$ is a sample quantity computed from the observed data. An alternative is to replace it with its population-level limit $\zeta(1) - \zeta(0) = 1 - p/2$, yielding the ratio estimator $\DeltaY^{(n)} / (1 - p/2)$. Both are consistent for GATE by \Cref{thm:consistency}. Using the sample denominator can reduce variance by adapting to finite-sample fluctuations in the realized exposure, provided the researcher's weights approximate the true weights. In addition, the theoretical compliance $\zeta(1) - \zeta(0)$ is straightforward to compute in the baseline setting but less so in the extensions of \Cref{sec:holdout,sec:multi-arm,sec:experiment-conflicts}, where holdout shares, the number of arms, or conflicting experiment assignments enter the expression. The empirical $\DeltaHt^{(n)}$ is therefore a simpler and more robust default: it tracks whatever assignment mechanism is in use, without a separate derivation for each variant of the design.

\subsection{Comparison with ERL}
\label{sec:erl-comparison}

\Cref{tab:simulation-mean-estimates} shows ERL attaining standard errors slightly below those of the Wald estimator. When many nodes have high edge-weight concentration, ERL can be more precise because it exploits the full randomization-induced variation in $H_i$, including variation due to neighbors' assignments, while Wald uses only the component associated with the instrument $D_i$. Conversely, it is easy to show that Wald's empirical first stage can cancel common exposure noise appearing in both the numerator and the denominator, which can make Wald more precise when treatment effects are homogeneous or nearly homogeneous.

However, this comparison largely understates both ERL's implementation burden and the standard errors it may incur under more complex designs. ERL requires each unit's exposure variance $v_i$, which is available in closed form under our baseline design. However, with a holdout group or concurrent experiments on the same surface, as in \Cref{sec:holdout,sec:experiment-conflicts}, a closed-form expression is generally unavailable. The associated assignment mechanisms can be complex and dynamic. Recovering $v_i$ then requires simulating the full assignment mechanism, which is difficult to run and may be infeasible when the graph is not retained after the experiment.

Moreover, treating the holdout group or concurrent experiments as random requires simulating the corresponding assignments as part of the mechanism, which may be intractable for the same reasons. Treating them as fixed instead removes the variance floor of \Cref{eq:vi-decomp}. This conditioning can leave $v_i$ at zero for units whose edges they fully determine (e.g., nodes connected only to holdout nodes) and inflate the weight of units that are only partly free. The Wald estimator does not suffer from these issues because it does not invert unit-level exposure variances, and its empirical first stage in \Cref{sec:denominator} tracks whatever mechanism is in force.

Finally, our setting does not generally admit the exact variance estimator of \citet{harshaw2023design}. Their Assumption~4, non-degenerate exposures, rules out degree-one nodes because such a node's exposure is binary. Degree-one nodes are important in social-network applications, and excluding them would largely change the target population.

\subsection{Two-sided platforms}
\label{sec:two-sided}

Many platforms are naturally two-sided, with edges connecting nodes of two distinct types, such as content creators and consumers, drivers and riders, buyers and sellers, hosts and guests, or couriers and eaters. Edge-level treatments can include any feature that affects the experience after a match is made, such as the order cancellation flow or in-trip support.

Formally, in this case the nodes of $G$ partition into two sides $\mathcal{A}$ and $\mathcal{B}$, and the graph is bipartite, so $w_{ij} = 0$ whenever $i$ and $j$ lie on the same side. Our assignment rule can be directly applied to this setting as well, as illustrated in \Cref{fig:two-sided}. However, each side now carries its own outcome metric and its own estimand, $\theta^{(n)}_{\mathcal{A}} = \frac{1}{|\mathcal{A}|}\sum_{i \in \mathcal{A}} b_i$ and $\theta^{(n)}_{\mathcal{B}} = \frac{1}{|\mathcal{B}|}\sum_{i \in \mathcal{B}} b_i$. We can apply the Wald estimator separately to the sampled nodes of each side, recovering $\theta^{(n)}_{\mathcal{A}}$ and $\theta^{(n)}_{\mathcal{B}}$. One can show that under \Cref{asm:sparse-weak} applied to the bipartite graph, both side estimators are consistent.

\begin{figure}[t]
    \centering
    \begin{tikzpicture}[
        sqsampled/.style={rectangle, draw, thick, minimum size=7mm, font=\bfseries},
        squnsampled/.style={rectangle, draw, fill=black!15, minimum size=7mm},
        cisampled/.style={circle, draw, thick, minimum size=7mm, font=\bfseries},
        ciunsampled/.style={circle, draw, fill=black!15, minimum size=7mm},
        cedge/.style={thick, blue!60!black},
        tedge/.style={thick, red!70!black},
        edgelabel/.style={fill=white, inner sep=1pt, font=\small},
    ]
        % top row (side A): squares
        \node[sqsampled]  (t1) at (0, 1.8) {C};
        \node[sqsampled]  (t2) at (2, 1.8) {T};
        \node[squnsampled] (t3) at (4, 1.8) {};
        % bottom row (side B): circles
        \node[ciunsampled] (b1) at (0, 0) {};
        \node[cisampled]  (b2) at (2, 0) {C};
        \node[cisampled]  (b3) at (4, 0) {T};
        % edges from top C (t1) to all three bottom nodes
        \draw[cedge] (t1) -- node[edgelabel] {C} (b1);
        \draw[cedge] (t1) -- node[edgelabel] {C} (b2);
        \draw[tedge] (t1) -- node[edgelabel] {T} (b3);
        % edges from bottom T (b3) to top T and top unsampled
        \draw[tedge] (b3) -- node[edgelabel] {T} (t2);
        \draw[tedge] (b3) -- node[edgelabel] {T} (t3);
    \end{tikzpicture}
    \caption{Edge assignments on a two-sided platform.}
    \label{fig:two-sided}
    \vspace{1ex}
    \begin{minipage}{\textwidth}
        \footnotesize \textit{Note:} Side $\mathcal{A}$ nodes are drawn as squares and side $\mathcal{B}$ nodes as circles. Sampled nodes ($S_i = 1$) are labeled by their treatment assignment $D_i \in \{C, T\}$, while unsampled nodes ($S_i = 0$) are shown as gray-filled. Edge labels show the realized edge treatment $Z_{ij}$ from \Cref{eq:assignment_rule}. The edge between the control square and the test circle is in conflict and has been drawn as $T$.
    \end{minipage}
\end{figure}

The motivation to apply our methodology is different in the two-sided case. In the bipartite graph setting, one could instead randomize only one side, say side $\mathcal{A}$, and let every edge inherit the treatment of its endpoint in $\mathcal{A}$. Because the graph is bipartite and only side $\mathcal{A}$ is randomized, each edge has exactly one randomized endpoint, so no conflicting assignments arise. The effect on side $\mathcal{A}$ can then be estimated by a simple difference in means, since each such node's edges all share its own treatment. The effect on side $\mathcal{B}$ can be estimated with the ERL estimator of \citet{harshaw2023design}. In the latter case, however, either clustering is needed or the exposure of high-degree nodes concentrates and inflates the variance of the estimator.

Our approach instead delivers a simple, low-variance estimator for both sides at once. The cost is that we rely on the linearity assumption (\Cref{asm:linearity}) for both sides, not only for side $\mathcal{B}$, and we assume away interference between nodes in $\mathcal{B}$ through nodes in $\mathcal{A}$ (\Cref{asm:locality}). See \citet{johari2022experimental, bajari2021multiple} on interference in two-sided platforms.

\section{Conclusion}
\label{sec:conclusion}

We have presented a framework for running and analyzing experiments with edge-level treatments in networks. By using node-level randomization as an instrument for edge-level exposure, the Wald estimator is consistent and asymptotically normal for GATE under a linear potential outcomes model and sparse graph assumptions.

Several directions remain for future work. \Cref{asm:locality} rules out spillovers from neighbors' exposures to a node's outcome. Extending the framework to accommodate such spillovers, perhaps through richer exposure mappings, is a natural next step. See \citet{hudgens2008toward, aronow2017estimating, saveski2017detecting, savje2024causal, leung2022causal} for different approaches to network interference.

Another assumption that may fail in practice is graph exogeneity (\Cref{asm:exogenous-graph}): in some applications the treatment may itself shift the edge weights. For instance, a new feature could change how a user distributes communications across friends. Characterizing the bias when weights are endogenous to treatment, and adjusting for it, is an open question. \citet{wang2026experimentation, gao2024endogenous} study experimentation when the treatment itself reshapes the graph.

Finally, extending the method from pairwise to group communication is a natural next step. In such settings, conflicting assignments introduce heterogeneous compliance: for example, users participating in larger groups face a higher chance of conflicting assignments and therefore have systematically lower compliance. The individual compliances can still be computed numerically from the assignment rule and the graph. Several approaches in the literature could then be adapted to recover GATE rather than a compliance-weighted average: 2SLS arguments that absorb compliance heterogeneity through treatment-by-covariate interactions \citep{wooldridge1997two,wooldridge2003further}, and weighting or matching-based IV estimators in the spirit of \citet{abadie2003semiparametric} and \citet{frolich2007nonparametric}, adapted so the weights neutralize per-unit compliance.

\printbibliography

\begin{appendices}
\crefalias{section}{appendix}
\crefalias{subsection}{appendix}

\section{Proofs}
\label{app:proofs}

The following concentration argument will be useful in several of the proofs below.

\begin{lemma}[Local averages concentrate]
\label{lma:local-average}
    Let $g^{(n)}$ be a sequence of polynomials whose coefficients are bounded
    uniformly in $n$, and set $g_i^{(n)} = g^{(n)}(S_i, D_i, H_i, \tilde H_i, Y_i)$. Under
    \Cref{asm:bounded,asm:random,asm:sparse-weak},
    \[
        \frac{1}{n} \sum_i \big(g_i^{(n)} - \E[g_i^{(n)}]\big) \toprob 0.
    \]
\end{lemma}

\begin{proof}
    Since $S_i, D_i \in \{0, 1\}$, $H_i, \tilde H_i \in [0, 1]$, $Y_i$ is bounded by
    \Cref{asm:bounded}, and the coefficients of $g^{(n)}$ are bounded,
    $g_i^{(n)}$ is uniformly bounded by a constant $C$. Each $g_i^{(n)}$
    is determined by $S_i$, $D_i$, and the variables $\{S_j, D_j, U_{ij}\}$ on
    edges of $G$ incident to $i$, so by \Cref{asm:random} the terms $g_i^{(n)}$ and
    $g_j^{(n)}$ are independent unless $i$ and $j$ coincide, are adjacent in $G$, or
    share a common neighbor in $G$. Only such pairs contribute to the covariance sum,
    and each variance and covariance is bounded by $C^2$, so
    \begin{align*}
        \var\left(\frac{1}{n} \sum_i g_i^{(n)}\right)
            &= \frac{1}{n^2}\left(\sum_i \var\big(g_i^{(n)}\big)
                + \sum_{i \ne j} \cov\big(g_i^{(n)}, g_j^{(n)}\big)\right) \\
            &< \frac{1}{n^2} \left(C^2 n
                + C^2 \left(\frac{\dmax}{2} n + \binom{\dmax}{2} n\right)\right)
            = O\left(\frac{\dmax^2}{n}\right),
    \end{align*}
    which converges to zero by \Cref{asm:sparse-weak}. Chebyshev's inequality
    gives $\frac{1}{n} \sum_i \big(g_i^{(n)} - \E[g_i^{(n)}]\big) \toprob 0$.
\end{proof}

\subsection{Proof of \texorpdfstring{\Cref{thm:ols}}{Theorem 1}}
\label{app:proof-ols}

\begin{proof}
    We write $v_i \coloneq \var(H_i \mid S_i = 1)$ for the conditional design variance of exposure. The OLS slope equals the ratio of the sample covariance to the sample variance of $H_i$ over the sampled units,
    \[
        \hat\beta_{\mathrm{OLS}}^{(n)} = \frac{\frac{1}{N} \sum_i S_i (H_i - \bar H)(Y_i - \bar Y)}{\frac{1}{N} \sum_i S_i (H_i - \bar H)^2},
        \qquad
        \bar H \coloneq \frac{1}{N} \sum_i S_i H_i, \quad \bar Y \coloneq \frac{1}{N} \sum_i S_i Y_i.
    \]
    By \Cref{lma:uniform-compliance}, $\E[H_i \mid S_i = 1, D_i = d] = \zeta(d)$ does not depend on $i$, so the conditional mean
    \[
        \mu \coloneq \E[H_i \mid S_i = 1] = \tfrac{1}{2}\big(\zeta(1) + \zeta(0)\big)
    \]
    is the same constant for every node $i$. Consequently $\E[H_i^2 \mid S_i = 1] = v_i + \mu^2$. By \Cref{lma:local-average},
    \begin{align*}
        \frac{1}{n} \sum_i S_i &\toprob p, &
        \frac{1}{n} \sum_i S_i H_i &\toprob p\mu, \\
        \frac{1}{n} \sum_i S_i H_i^2 &\toprob \frac{p}{n} \sum_i (v_i + \mu^2), &
        \frac{1}{n} \sum_i S_i Y_i &\toprob \frac{p}{n} \sum_i (a_i + b_i \mu), \\
        \frac{1}{n} \sum_i S_i H_i Y_i &\toprob \frac{p}{n} \sum_i \big(a_i \mu + b_i(v_i + \mu^2)\big), &&
    \end{align*}
    where the last two use $Y_i = a_i + b_i H_i$ from \Cref{asm:linearity}.

    We identify the probability limits of the numerator and denominator of $\hat\beta_{\mathrm{OLS}}^{(n)}$ separately and combine them with Slutsky's theorem. For the denominator,
    \[
        \frac{1}{N} \sum_i S_i (H_i - \bar H)^2 = \frac{1}{N} \sum_i S_i (H_i^2 - \bar H^2) = \frac{\frac{1}{n} \sum_i S_i H_i^2}{\frac{1}{n} \sum_i S_i} - \left( \frac{\frac{1}{n} \sum_i S_i H_i}{\frac{1}{n} \sum_i S_i} \right)^2.
    \]
    Slutsky's theorem and the limits above give
    \[
        \frac{1}{N} \sum_i S_i (H_i - \bar H)^2 \toprob \frac{\frac{p}{n}\sum_i (v_i + \mu^2)}{p} - \mu^2 = \frac{1}{n} \sum_i v_i.
    \]
    Analogously,
    \[
        \frac{1}{N} \sum_i S_i (H_i - \bar H)(Y_i - \bar Y) = \frac{1}{N} \sum_i S_i (H_i Y_i - \bar H \bar Y) = \frac{\frac{1}{n} \sum_i S_i H_i Y_i}{\frac{1}{n} \sum_i S_i} - \frac{\left(\frac{1}{n} \sum_i S_i H_i\right)\left(\frac{1}{n} \sum_i S_i Y_i\right)}{\left(\frac{1}{n} \sum_i S_i\right)^2},
    \]
    and Slutsky's theorem yields
    \begin{align*}
        \frac{1}{N} \sum_i S_i (H_i - \bar H)(Y_i - \bar Y)
        &\toprob \frac{\frac{p}{n}\sum_i\big(a_i \mu + b_i(v_i + \mu^2)\big)}{p}
        - \frac{(p\mu)\big(\frac{p}{n}\sum_i(a_i + b_i\mu)\big)}{p^2} \\
        &= \mu \bar a + \frac{1}{n}\sum_i b_i v_i + \mu^2 \bar b - \mu\big(\bar a + \mu \bar b\big) = \frac{1}{n} \sum_i b_i v_i.
    \end{align*}

    The denominator limit $\frac{1}{n}\sum_i v_i$ is bounded away from zero, as we show below in \Cref{eq:vi-decomp}, so a final application of Slutsky's theorem gives
    \[
        \hat\beta_{\mathrm{OLS}}^{(n)} - \frac{\frac{1}{n} \sum_i b_i v_i}{\frac{1}{n} \sum_i v_i} \toprob 0.
    \]
    Writing $\bar v = \frac{1}{n}\sum_i v_i$, the limit decomposes as
    \[
        \frac{\sum_i b_i v_i}{\sum_i v_i} = \bar b + \frac{\frac{1}{n}\sum_i (b_i - \bar b)(v_i - \bar v)}{\bar v},
    \]
    so it equals $\thetan = \bar b$ if and only if the empirical covariance between $b_i$ and $v_i$ vanishes in the limit.

    Finally, we record the explicit form of $v_i$. By the law of total variance over $D_i$,
    \[
        v_i = \var_{D_i}\!\big(\E[H_i \mid S_i = 1, D_i]\big) + \E_{D_i}\!\big[\var(H_i \mid S_i = 1, D_i)\big].
    \]
    For the first term, $\E[H_i \mid S_i = 1, D_i = d] = \zeta(d)$ by \Cref{lma:uniform-compliance}, so it equals $\frac{1}{4}\big(\zeta(1) - \zeta(0)\big)^2$. For the second term, conditional on $D_i = d$ the edge treatments $\{Z_{ij}\}_j$ are independent across $j$ by \Cref{asm:random}. Each $Z_{ij} \in \{0, 1\}$ is binary with conditional mean $\E[Z_{ij} \mid S_i = 1, D_i = d] = \zeta(d)$, so $\var(Z_{ij} \mid S_i = 1, D_i = d) = \zeta(d)\big(1 - \zeta(d)\big)$ and
    \[
        \var(H_i \mid S_i = 1, D_i = d) = \sum_j w_{ij}^2\, \var(Z_{ij} \mid S_i = 1, D_i = d) = \zeta(d)\big(1 - \zeta(d)\big) \sum_j w_{ij}^2.
    \]
    Combining the two terms,
    \begin{equation}
        v_i = \frac{1}{4}\big(\zeta(1) - \zeta(0)\big)^2 + \kappa \sum_j w_{ij}^2,
    \label{eq:vi-decomp}
    \end{equation}
    where $\kappa \coloneq \big[\zeta(1)\big(1 - \zeta(1)\big) + \zeta(0)\big(1 - \zeta(0)\big)\big]/2$.
    The second term in \Cref{eq:vi-decomp} is increasing in the weight concentration $\sum_j w_{ij}^2$.
\end{proof}

\subsection{Proof of \texorpdfstring{\Cref{thm:consistency}}{Theorem 2}}
\label{app:proof-consistency}

\begin{proof}
    We first show that $\frac{1}{N_1} \sum_i Y_i S_i D_i - \frac{1}{n} \sum_i \E[Y_i \mid S_i = 1, D_i = 1] \toprob 0$. Writing the sample mean as a ratio of two averages,
    \[
        \frac{1}{N_1} \sum_i Y_i S_i D_i = \frac{\frac{1}{n} \sum_i Y_i S_i D_i}{\frac{1}{n} \sum_i S_i D_i},
    \]
    we treat the numerator and denominator separately. By \Cref{lma:local-average}, $\frac{1}{n} \sum_i Y_i S_i D_i - \frac{1}{n} \sum_i \E[Y_i S_i D_i] \toprob 0$. The denominator converges in probability to $\E[S_i D_i] = p/2 > 0$ by the law of large numbers. Combined with the numerator's convergence and the uniform bound $\left|\frac{1}{n} \sum_i \E[Y_i S_i D_i]\right| \leqslant M \cdot p/2$ from \Cref{asm:bounded}, Slutsky's theorem yields
    \[
        \frac{1}{N_1} \sum_i Y_i S_i D_i - \frac{\frac{1}{n} \sum_i \E[Y_i S_i D_i]}{\E[S_i D_i]} \toprob 0.
    \]
    The limit on the right simplifies as
    \[
        \frac{\frac{1}{n} \sum_i \E[Y_i S_i D_i]}{\E[S_i D_i]} = \frac{1}{n} \sum_i \frac{\E[Y_i S_i D_i]}{\PP[S_i = 1, D_i = 1]} = \frac{1}{n} \sum_i \E[Y_i \mid S_i = 1, D_i = 1],
    \]
    yielding $\frac{1}{N_1} \sum_i Y_i S_i D_i - \frac{1}{n} \sum_i \E[Y_i \mid S_i = 1, D_i = 1] \toprob 0$. The same argument, applied with $\tilde H_i$ in place of $Y_i$ and with $1 - D_i$ in place of $D_i$ for the control side, yields the three analogous statements for $\frac{1}{N_0} \sum_i Y_i S_i (1-D_i)$, $\frac{1}{N_1} \sum_i \tilde H_i S_i D_i$, and $\frac{1}{N_0} \sum_i \tilde H_i S_i (1-D_i)$. Combining the two $\tilde H_i$ statements gives $\DeltaHt^{(n)} - (\zeta(1) - \zeta(0)) \toprob 0$, and since $\zeta(1) - \zeta(0) = 1 - p/2$ is nonzero, Slutsky's theorem applied to $\thetahat = \DeltaY^{(n)}/\DeltaHt^{(n)}$ yields that $\thetahat$ converges in probability to the ratio of the limits. By \Cref{asm:linearity} and \Cref{lma:uniform-compliance}, this ratio simplifies as
    \begin{align*}
        &~ \frac{\frac{1}{n} \sum_i \E[Y_i \mid S_i = 1, D_i = 1] - \frac{1}{n} \sum_i \E[Y_i \mid S_i = 1, D_i = 0]}{\frac{1}{n} \sum_i \E[\tilde H_i \mid S_i = 1, D_i = 1] - \frac{1}{n} \sum_i \E[\tilde H_i \mid S_i = 1, D_i = 0]} \\
        =&~ \frac{\frac{1}{n} \sum_i (a_i + b_i \E[H_i \mid S_i = 1, D_i = 1]) - \frac{1}{n} \sum_i (a_i + b_i \E[H_i \mid S_i = 1, D_i = 0])}{\frac{1}{n} \sum_i \E[\tilde H_i \mid S_i = 1, D_i = 1] - \frac{1}{n} \sum_i \E[\tilde H_i \mid S_i = 1, D_i = 0]} \\
        =&~ \sum_i b_i \cdot \frac{\E[H_i \mid S_i = 1, D_i = 1] - \E[H_i \mid S_i = 1, D_i = 0]}{\sum_j (\E[\tilde H_j \mid S_j = 1, D_j = 1] - \E[\tilde H_j \mid S_j = 1, D_j = 0])} \\
        =&~ \sum_i b_i \cdot \frac{\zeta(1) - \zeta(0)}{\sum_j (\zeta(1) - \zeta(0))} = \bar{b},
    \end{align*}
    where the $a_i$ terms cancel and \Cref{lma:uniform-compliance} makes the compliance constant in $i$ for both the true and proxy exposures, so the weights are uniform, proving $\thetahat - \thetan \toprob 0$.
\end{proof}

\subsection{Proof of \texorpdfstring{\Cref{thm:nonlinear-weighting}}{Theorem 3}}
\label{app:proof-nonlinear-weighting}

\begin{proof}
    Applying the first part of the proof of \Cref{thm:consistency} with $Y_i=y_i(H_i)$, using the uniform boundedness of $y_i$, gives
    \[
        \DeltaY^{(n)}
        -\frac{1}{n}\sum_i
        \E\left[
            y_i\big(H_i(1)\big)-y_i\big(H_i(0)\big)
            \;\middle|\;S_i=1
        \right]
        \toprob0.
    \]
    The same proof already establishes $\DeltaHt^{(n)}-\big(\zeta(1)-\zeta(0)\big)\toprob0$.
    Slutsky's theorem therefore yields
    \begin{equation}
        \thetahat
        -\frac{1}{\zeta(1)-\zeta(0)}
        \frac{1}{n}\sum_i\E\left[
            y_i\big(H_i(1)\big)-y_i\big(H_i(0)\big)
            \;\middle|\;S_i=1
        \right]
        \toprob0.
    \label{eq:nonlinear-secant-limit}
    \end{equation}
    For a sampled node, define the potential edge treatment $Z_{ij}(d)$ by setting $D_i=d$ and leaving all other assignments and edge tie-breaks random. The pair $(Z_{ij}(0),Z_{ij}(1))$ equals $(0,1)$ when $S_j=0$, $(0,U_{ij})$ when $S_j=1,D_j=0$, and $(U_{ij},1)$ when $S_j=1,D_j=1$. Hence $Z_{ij}(1)\geqslant Z_{ij}(0)$ for every realization. Since the true edge weights are nonnegative, $H_i(1)\geqslant H_i(0)$.
    Absolute continuity of $y_i$ and the monotonicity $H_i(1)\geqslant H_i(0)$ imply
    \[
        y_i\big(H_i(1)\big)-y_i\big(H_i(0)\big)
        =\int_0^1y_i'(h)
        \mathds{1}\{H_i(0)<h\leqslant H_i(1)\}\,dh.
    \]
    Taking expectations and interchanging expectation and integration in \Cref{eq:nonlinear-secant-limit} gives \Cref{eq:nonlinear-acr,eq:nonlinear-weight}. In addition, uniform compliance (\Cref{lma:uniform-compliance}) gives
    \[
        \int_0^1\omega_i(h)\,dh
        =\frac{
            \E[H_i(1)-H_i(0)\mid S_i=1]
        }{\zeta(1)-\zeta(0)}
        =1.
    \]
\end{proof}

\subsection{Proof of \texorpdfstring{\Cref{lma:linearize}}{Lemma 2}}
\label{app:proof-linearize}

\begin{proof}
    The numerator in \Cref{eq:linear} can be expressed as an average
    \begin{align*}
        \DeltaY^{(n)} - \thetan \DeltaHt^{(n)} &= \frac{1}{N_1} \sum_i (Y_i - \thetan \tilde H_i) S_i D_i - \frac{1}{N_0} \sum_i (Y_i - \thetan \tilde H_i) S_i (1 - D_i) \\
        &= \frac{1}{N_1} \sum_i (Y_i - \thetan \tilde H_i - \abarn) S_i D_i - \frac{1}{N_0} \sum_i (Y_i - \thetan \tilde H_i - \abarn) S_i (1 - D_i) \\
        &= \frac{1}{n} \sum_i S_i \left( \frac{D_i}{N_1/n} - \frac{1-D_i}{N_0/n} \right) (Y_i - \thetan \tilde H_i - \abarn).
    \end{align*}
    We center each term by $\abarn$ intentionally, following the usual analysis for Hájek estimators.
    Let $X_i \coloneq Y_i - \thetan \tilde H_i - \abarn$ be the unscaled error for node $i$. We can rewrite the numerator terms of our estimator as sample averages. For the test group, we define the average transformed outcome $\Xbarn \coloneq \frac{1}{n}\sum_i S_i D_i X_i$ and the sample size proportion $\phatn \coloneq N_1 / n$. Analogous quantities $\bar X_0^{(n)}$ and $\hat p_0^{(n)}$ are defined for the control group. The linearized numerator can then be written as
    \[
        \DeltaY^{(n)} - \thetan \DeltaHt^{(n)} = \frac{\Xbarn}{\phatn} - \frac{\bar X_0^{(n)}}{\hat p_0^{(n)}}.
    \]
    We analyze the test ratio $\Xbarn/\phatn$ in detail. The control ratio $\bar X_0^{(n)}/\hat p_0^{(n)}$ is treated identically by symmetry, and the two are combined at the end.

    We first calculate the finite-$n$ expectations of the test elements. For the proportion, we have $\E[\phatn] = p/2$. For the average transformed outcome, we have
    \begin{align*}
        \E[\Xbarn] &= \frac{p}{2} \cdot \frac{1}{n} \sum_i \left( a_i + b_i \E[H_i \mid S_i=1, D_i=1] - \thetan \E[\tilde H_i \mid S_i=1, D_i=1] - \abarn \right) = 0,
    \end{align*}
    since $\thetan = \bar{b}^{(n)}$ and, by \Cref{lma:uniform-compliance}, $\E[H_i \mid S_i=1, D_i=1] = \E[\tilde H_i \mid S_i=1, D_i=1] = \zeta(1)$ is constant across all units. This is expected, since the $X_i$ are centered.

    Applying the delta method to $\Xbarn/\phatn$, we expand the function $(x, y) \mapsto x / y$ around the finite-$n$ expectations $(0, p/2)$, yielding
    \begin{align*}
        \frac{\Xbarn}{\phatn} &= 0 + \frac{1}{p/2}\left(\Xbarn - 0\right) - \frac{0}{(p/2)^2}\left(\phatn - \frac{p}{2}\right) + R_1^{(n)} = \frac{2}{p} \Xbarn + R_1^{(n)}.
    \end{align*}
    Substituting the definition of $\Xbarn$ back into the linear term gives
    \[
        \frac{\Xbarn}{\phatn} = \frac{1}{n} \sum_i \frac{2}{p} S_i D_i X_i + R_1^{(n)}.
    \]

    The linearization in the lemma requires the remainder $R_1^{(n)}$ to be $o_p(n^{-1/2})$. By Taylor's theorem, the remainder $R_1^{(n)}$ is a quadratic function of $\left(\Xbarn, \phatn - \frac{p}{2}\right)$ with coefficients proportional to entries of the Hessian matrix of $(x, y) \mapsto x/y$. We work on the event $\{\phatn > p/4\}$, whose complement has probability $O(e^{-cn})$ by Hoeffding's inequality and is therefore negligible for the $o_p(n^{-1/2})$ statement. There the Taylor intermediate point has second coordinate bounded below by $p/4$, so all Hessian entries of $f(x,y) = x/y$ are bounded (in particular $|f_{xy}| = 1/y^2 < 16/p^2$), and $\Xbarn = o_p(1)$ by \Cref{asm:bounded}. The $(\Xbarn)^2$ contribution vanishes outright because $\partial^2 f/\partial x^2 \equiv 0$. The $(\phatn - p/2)^2$ contribution is $o_p(n^{-1/2})$, since $\phatn$ is an average of i.i.d.\ $\Bern(p/2)$ indicators and so $\phatn - p/2 = O_p(n^{-1/2})$. The leading remainder is therefore the cross product $\Xbarn(\phatn - p/2)$. It suffices to show $\E\!\left|\Xbarn(\phatn - p/2)\right| = o(n^{-1/2})$.

    By Cauchy--Schwarz, $\E\!\left|\Xbarn(\phatn - p/2)\right| \leqslant \sqrt{\var(\Xbarn) \var(\phatn)}$. By the proof of \Cref{lma:local-average}, $\var(\Xbarn) = O(\dmax^2/n)$, and $\var(\phatn) = O(n^{-1})$ since $\phatn$ is an average of i.i.d.\ $\Bern(p/2)$ indicators. Thus the cross term is $O(\dmax/n) = o(n^{-1/2})$ under \Cref{asm:sparse-weak} ($\dmax = o(n^{1/2})$), and so $R_1^{(n)} = o_p(n^{-1/2})$.

    Following the same calculation for the control group and taking a difference yields
    \[
        \DeltaY^{(n)} - \thetan\DeltaHt^{(n)} = \frac{1}{n} \sum_i \frac{2}{p} S_i (2D_i - 1) X_i + o_p(n^{-1/2}) = \frac{1}{n}\sum_i \psi_i + o_p(n^{-1/2}). \qedhere
    \]
\end{proof}

\subsection{Proof of \texorpdfstring{\Cref{thm:normal}}{Theorem 4}}
\label{app:proof-normal}

\begin{proof}
    We start with the asymptotic normality of the numerator, as linearized in \Cref{lma:linearize}, $\DeltaY^{(n)} - \thetan\DeltaHt^{(n)} = \frac{1}{n} \sum_i \psi_i + o_p(n^{-1/2})$. It suffices to establish the asymptotic normality of $\frac{1}{n} \sum_i \psi_i$, and we apply Stein's method for local dependence \citep{stein1986approximate, baldi1989normal}.

    Recall that $\psi_i = \frac{2}{p}S_i(2D_i - 1)X_i$ where $X_i = Y_i - \thetan \tilde H_i - \abarn$. Substituting the linear model $Y_i = a_i + b_i H_i$ and expanding the exposure terms $H_i = \sum_j w_{ij}Z_{ij}$ and $\tilde H_i = \sum_j \tilde w_{ij}Z_{ij}$, we can decompose the node-level sum $\sum_i \psi_i$ into a sum over independent node baseline terms and undirected edge interaction terms. Let $V = \{N_i: i = 1, \ldots, n\} \cup \{E_{ij}: i < j\}$ be a combined set of random variables defined as
    \begin{align*}
        N_i &= \frac{2}{p}S_i(2D_i - 1)(a_i - \abarn), \\
        E_{ij} &= \frac{2}{p} Z_{ij}\left[S_i(2D_i - 1)(b_i w_{ij} - \thetan \tilde w_{ij}) + S_j(2D_j - 1)(b_j w_{ji} - \thetan \tilde w_{ji})\right],
    \end{align*}
    such that $\sum_i \psi_i = \sum_i N_i + \sum_{i < j} E_{ij}$. The complete sum is centered because
    \[
        \sum_i \E[\psi_i]
        = \left(\zeta(1)-\zeta(0)\right)\sum_i(b_i-\bar b)
        = 0.
    \]
    Define $N_i^\circ \coloneq N_i-\E[N_i]$ and $E_{ij}^\circ \coloneq E_{ij}-\E[E_{ij}]$. Then
    \[
        \sum_i N_i^\circ + \sum_{i<j}E_{ij}^\circ
        = \sum_iN_i + \sum_{i<j}E_{ij}
        = \sum_i\psi_i.
    \]

    We analyze the dependency graph of $V$. The node term $N_i$ depends only on the assignments $S_i$ and $D_i$. The edge term $E_{ij}$ depends on $S_i, D_i, S_j, D_j$, and the edge randomness $U_{ij}$. Consequently, $E_{ij}$ is independent of all other variables except $N_i, N_j$, and any other edge term incident to node $i$ or node $j$. The maximum degree of this combined dependency graph is at most $2\dmax$.

    To establish asymptotic normality, we apply Stein's method as stated in Theorem 1 of \citet{baldi1989normal}. Let $L(k)$ denote the maximum number of connected subgraphs of size $k$ containing a given vertex in the dependency graph. Since the maximum degree is at most $2\dmax$, we have $L(3) = O(\dmax^2)$ and $L(4) = O(\dmax^3)$.

    We must bound the sums of the absolute third and fourth centered moments over all $v \in V$. By \Cref{asm:bounded}, the potential outcomes $a_i$ and $b_i$ are uniformly bounded by $M$, so $\E[|N_i|^3] = O(1)$ and $\E[N_i^4] = O(1)$. For the edge terms, $|E_{ij}| \leqslant C(w_{ij} + \tilde w_{ij} + w_{ji} + \tilde w_{ji})$ for some constant $C$. Both weightings satisfy $\sum_j w_{ij} = \sum_j \tilde w_{ij} = 1$ and lie in $[0,1]$, so $w_{ij}^3 \leqslant w_{ij}$, $\tilde w_{ij}^3 \leqslant \tilde w_{ij}$, and likewise for the fourth powers. Therefore, the sum of the edge moments scales as
    \[
        \sum_{i < j} \E[|E_{ij}|^3] \leqslant 16C^3 \sum_{i < j} (w_{ij}^3 + \tilde w_{ij}^3 + w_{ji}^3 + \tilde w_{ji}^3) = 16C^3 \sum_{i,j} (w_{ij}^3 + \tilde w_{ij}^3) \leqslant 16C^3 \sum_{i,j} (w_{ij} + \tilde w_{ij}) = 32 C^3 n,
    \]
    where the first inequality uses Jensen's inequality $\left(\sum_{k=1}^m x_k\right)^q \leqslant m^{q-1} \sum_{k=1}^m x_k^q$ for nonnegative $x_k$. The exact same bounding logic holds for the fourth moments, giving
    \[
        \sum_i \E[|N_i|^3] + \sum_{i<j} \E[|E_{ij}|^3] = O(n), \quad \sum_i \E[|N_i|^4] + \sum_{i<j} \E[|E_{ij}|^4] = O(n).
    \]
    For $q \in \{3,4\}$, Jensen's inequality gives
    \[
        \E[|X-\E[X]|^q] \leqslant 2^q \E[|X|^q].
    \]
    Hence the centered moments of $N_i^\circ$ and $E_{ij}^\circ$ have the same orders. Centering does not change the dependency graph.

    By Stein's bound, the uniform distance between the distribution of the standardized sum $\frac{1}{\sqrt{n}\sigman} \sum_i \psi_i$ and the standard normal $\Normal(0, 1)$ is bounded by $c(T_3 + T_4)$, where
    \begin{align*}
        T_3 &= \sqrt{\frac{L(3)}{(\sqrt{n}\sigman)^3} \left(\sum_i \E[|N_i^\circ|^3] + \sum_{i<j} \E[|E_{ij}^\circ|^3]\right)}, \\
        T_4 &= \left(\frac{L(4)}{(\sqrt{n}\sigman)^4} \left(\sum_i \E[|N_i^\circ|^4] + \sum_{i<j} \E[|E_{ij}^\circ|^4]\right)\right)^{1/4}.
    \end{align*}
    Since $\sigman^2$ is bounded away from zero by \Cref{asm:no-superefficient}, substituting the moment bounds and dependency-graph rates gives
    \[
        T_3 = \sqrt{\frac{O(\dmax^2)}{O(n^{3/2})} O(n)} = O\!\left(\frac{\dmax}{n^{1/4}}\right), \qquad T_4 = \left(\frac{O(\dmax^3)}{O(n^2)} O(n)\right)^{1/4} = O\!\left(\frac{\dmax^{3/4}}{n^{1/4}}\right).
    \]
    Under \Cref{asm:sparse-medium}, $\dmax = o(n^{1/4})$, so both $T_3$ and $T_4$ converge to zero as $n \to \infty$.
    This application of Stein's method yields, via $\sqrt{n}(\DeltaY^{(n)} - \thetan\DeltaHt^{(n)}) = \frac{1}{\sqrt{n}}\sum_i \psi_i + o_p(1)$ from \Cref{lma:linearize},
    \[
        \frac{\sqrt{n}(\DeltaY^{(n)} - \thetan\DeltaHt^{(n)})}{\sigman} \todist \Normal(0, 1).
    \]
    Applying Slutsky's theorem allows us to divide the sequence by $\DeltaHt^{(n)}$ and the limit by $1 - p/2$, yielding $\frac{\sqrt{n}(\thetahat - \thetan)}{\sigman / (1 - p/2)} \todist \Normal(0, 1)$.
\end{proof}

\subsection{Proof of \texorpdfstring{\Cref{thm:variance-limit}}{Theorem 5}}
\label{app:proof-variance-limit}

\begin{proof}
    We define an oracle estimator $\sigmantilde^2 = \frac{1}{n} \sum_i \sum_{j \in N_2(i)} \psi_i \psi_j$ that uses the true parameters $\thetan$ and $\abarn$. We decompose the estimation error into three components
    \[
        \sigmanhat^2 - \sigman^2 - \Gamma_n
        = \underbrace{(\sigmanhat^2 - \sigmantilde^2)}_{\text{approximation error}}
        + \underbrace{(\sigmantilde^2 - \E[\sigmantilde^2])}_{\text{sampling error}}
        + \underbrace{(\E[\sigmantilde^2] - \sigman^2 - \Gamma_n)}_{\text{expectation gap}}.
    \]

    We first show that the expectation gap is zero. From the definition of $\psi_i$,
    \begin{align*}
        \E[\psi_i] &= \E[Y_i-\bar b\tilde H_i-\abarn \mid S_i=1,D_i=1] - \E[Y_i-\bar b\tilde H_i-\abarn \mid S_i=1,D_i=0] \\
        &= \left(\zeta(1)-\zeta(0)\right)(b_i-\bar b),
    \end{align*}
    where the second equality follows from \Cref{asm:linearity,lma:uniform-compliance}. Since $\psi_i$ and $\psi_j$ are independent when $j \notin N_2(i)$,
    \[
        \sigman^2
        = \frac{1}{n}\sum_i\sum_{j\in N_2(i)}
        \cov(\psi_i,\psi_j).
    \]
    Therefore,
    \begin{align*}
        \E[\sigmantilde^2]-\sigman^2 &= \frac{1}{n}\sum_i\sum_{j\in N_2(i)}\left(\E[\psi_i\psi_j]-\cov(\psi_i,\psi_j)\right) = \frac{1}{n}\sum_i\sum_{j\in N_2(i)}\E[\psi_i]\E[\psi_j] \\
        &= \left(\zeta(1)-\zeta(0)\right)^2 \frac{1}{n}\sum_i\sum_{j\in N_2(i)}(b_i-\bar b)(b_j-\bar b) = \Gamma_n.
    \end{align*}
    Thus the expectation gap is zero.

    We next address the approximation error, $\sigmanhat^2 - \sigmantilde^2$. We denote the estimation error for the linearized term for a single node as
    \[
        \delta_i \coloneq \hat{\psi}_i - \psi_i = \frac{2}{p}S_i(2D_i - 1)[-(\thetahat - \thetan)\tilde H_i - (\hat{a} - \abarn)].
    \]

    To establish the convergence rate of $\delta_i$, we first derive the rate $\thetahat - \thetan = O_p(\dmax/\sqrt n)$ without invoking \Cref{thm:normal} (and hence \Cref{asm:no-superefficient}). Chebyshev's inequality applied to \Cref{lma:linearize}, using the local-dependence variance bound from the proof of \Cref{lma:local-average}, gives $\DeltaY^{(n)} - \thetan\DeltaHt^{(n)} = O_p(\dmax/\sqrt n)$. Combined with $\DeltaHt^{(n)} \toprob \zeta(1)-\zeta(0)=1-p/2>0$ from \Cref{thm:consistency}, this yields the claim.

    We then bound the two terms in the bracket $-(\thetahat - \thetan)\tilde H_i - (\hat{a} - \abarn)$ separately. The first contributes $O_p(\dmax/\sqrt{n})$ by the rate above and the boundedness of $\tilde H_i$. For the second, we must bound the error of the sample baseline mean $\hat{a} = \frac{1}{N} \sum_i S_i(Y_i - \thetahat \tilde H_i)$. Substituting the linear model $Y_i = a_i + b_i H_i$ and the identity $\thetan = \bar{b}$, we factor out the random prefactor $n/N$ and decompose the estimator error as
    \begin{align*}
        \hat{a} - \abarn = \frac{n}{N} \Bigg[ & \frac{1}{n} \sum_i S_i (a_i - \abarn) + \frac{1}{n} \sum_i S_i (b_i - \thetan) H_i \\
        & + \frac{1}{n} \sum_i S_i \thetan (H_i - \tilde H_i) + \frac{1}{n} \sum_i S_i (\thetan - \thetahat) \tilde H_i \Bigg].
    \end{align*}
    Since $N/n \toprob p > 0$, the prefactor $n/N$ is $O_p(1)$, and the rates are governed by the four inner averages:
    \begin{enumerate}
        \item For the first, because $S_i$ are independent Bernoulli variables and $a_i$ are uniformly bounded by \Cref{asm:bounded}, its variance scales as $O(1/n)$. By Chebyshev's inequality, this term is $O_p(1/\sqrt{n})$.
        \item For the second, \Cref{lma:uniform-compliance} implies the expectation of the sum $\sum_i (b_i - \bar{b}) \PP(S_i = 1) \E[H_i \mid S_i = 1]$ is zero. By the proof of \Cref{lma:local-average}, the variance is $O(\dmax^2 / n)$ and this term is $O_p(\dmax/\sqrt{n})$.
        \item For the third, \Cref{lma:uniform-compliance} gives $\E[H_i \mid S_i = 1] = \E[\tilde H_i \mid S_i = 1]$, so the sum has mean zero. Since $\thetan = \bar b$ is bounded by \Cref{asm:bounded}, the proof of \Cref{lma:local-average} gives variance $O(\dmax^2 / n)$, so this term is $O_p(\dmax/\sqrt{n})$.
        \item For the fourth, the rate $\thetahat - \thetan = O_p(\dmax/\sqrt n)$ derived above and the boundedness of $\frac{1}{n}\sum_i S_i \tilde H_i$ give $O_p(\dmax/\sqrt n)$.
    \end{enumerate}

    Because $\dmax$ is allowed to grow with $n$, the second and third terms dominate, giving $\hat{a} - \abarn = O_p(\dmax/\sqrt{n})$. This dominates the first term as well, so
    \[
        \max_i |\delta_i| = O_p(\dmax/\sqrt{n}).
    \]

    We can now expand the difference of the products as $\hat{\psi}_i \hat{\psi}_j - \psi_i \psi_j = \delta_i \delta_j + \psi_i \delta_j + \delta_i \psi_j$, and rewrite the approximation error as
    \[
        \sigmanhat^2 - \sigmantilde^2 = \frac{1}{n} \sum_i \sum_{j \in N_2(i)} \left( \delta_i \delta_j + \psi_i \delta_j + \delta_i \psi_j \right).
    \]
    Since the influence terms $\psi_i$ are uniformly bounded and there are $O(n\dmax^2)$ ordered pairs with $j \in N_2(i)$, the triangle inequality gives
    \begin{align*}
        |\sigmanhat^2 - \sigmantilde^2| &\leqslant \frac{1}{n} \sum_i \sum_{j \in N_2(i)} \left(|\delta_i\delta_j| + |\psi_i\delta_j| + |\delta_i\psi_j|\right) = O_p\left(\frac{\dmax^4}{n} + \frac{\dmax^3}{\sqrt{n}}\right) = o_p(1),
    \end{align*}
    under \Cref{asm:sparse-strong}. Therefore, the approximation error vanishes, yielding $\sigmanhat^2 - \sigmantilde^2 \toprob 0$.

    Next we address the concentration of the oracle estimator, $\sigmantilde^2 - \E[\sigmantilde^2]$. The variance of this estimator expands into a sum of covariances between pairs in the dependency graph
    \[
        \var(\tilde{\sigma}_n^2) = \frac{1}{n^2} \sum_i \sum_{j \in N_2(i)} \sum_{k=1}^n \sum_{l \in N_2(k)} \cov(\psi_i \psi_j, \psi_k \psi_l).
    \]

    For the covariance term to be nonzero, the dependency neighborhoods of the two pairs must overlap, meaning $(N_1(i) \cup N_1(j)) \cap (N_1(k) \cup N_1(l)) \ne \emptyset$. When this happens, the covariance term is bounded because of \Cref{asm:bounded}. There are $O(n\dmax^2)$ choices for the local pair $(i,j)$. For each pair, choosing an overlapping node, a nearby endpoint $k$, and $l \in N_2(k)$ gives at most $O(\dmax^4)$ dependent pairs $(k,l)$, so the number of overlapping 4-tuples is bounded by $O(n\dmax^6)$. By \Cref{asm:sparse-strong}, $\dmax = o(n^{1/6})$, so $\dmax^6 = o(n)$ and the count is $o(n^2)$. Summing over this restricted set of overlapping dependencies yields
    \[
        \var(\tilde{\sigma}_n^2) \leqslant \frac{1}{n^2} \cdot o(n^2) = o(1).
    \]
    By Chebyshev's inequality, we have $\tilde{\sigma}_n^2 - \E[\tilde{\sigma}_n^2] \toprob 0$.

    Because the expectation gap is zero and both the approximation error and the sampling error vanish, we conclude that $\sigmanhat^2 - \sigman^2 - \Gamma_n \toprob 0$.
\end{proof}

\subsection{Proof of \texorpdfstring{\Cref{cor:conservative-ci}}{the conservative confidence interval corollary}}
\label{app:proof-conservative-ci}

\begin{proof}
    We can write the absolute value of the empirical test statistic as
    \[
        \left|\frac{\sqrt{n}(\thetahat-\thetan)}{\sigmanhat/\DeltaHt^{(n)}}\right|
        = \left|\frac{\sqrt{n}(\thetahat-\thetan)}{\sigman/(1-p/2)}\right|
        \frac{\sigman}{\sigmanhat}
        \left|\frac{\DeltaHt^{(n)}}{1-p/2}\right|.
    \]
    By \Cref{thm:normal}, the first term converges in distribution to the absolute value of a standard normal random variable. By \Cref{asm:no-superefficient}, there exists $c>0$ such that $\sigman^2\geqslant c$ for all sufficiently large $n$. Fix $\varepsilon>0$ and choose $\eta>0$ such that $2\eta<c[1-(1+\varepsilon)^{-2}]$. By \Cref{asm:nonnegative-local-alignment}, $\Gamma_n\geqslant-\eta$ for all sufficiently large $n$, while \Cref{thm:variance-limit} implies $\sigmanhat^2\geqslant\sigman^2+\Gamma_n-\eta$ with probability approaching one. Hence
    \[
        \sigmanhat^2\geqslant\frac{\sigman^2}{(1+\varepsilon)^2},
    \]
    with probability approaching one, which gives
    \[
        \PP\left(\frac{\sigman}{\sigmanhat}>1+\varepsilon\right)\to0.
    \]
    Finally, \Cref{thm:consistency} gives $\DeltaHt^{(n)} / (1-p/2) \toprob 1$. Therefore, for every $\varepsilon>0$, the absolute value of the empirical test statistic is bounded above by $(1+\varepsilon)$ times the asymptotically standard normal statistic with probability approaching one. Its asymptotic rejection probability at the critical value $z_{1-\alpha/2}$ is therefore bounded above by $\PP(|Z|>z_{1-\alpha/2}/(1+\varepsilon))$ for $Z \sim \Normal(0, 1)$. Letting $\varepsilon \downarrow 0$ gives an upper bound of $\alpha$ on the rejection probability, or equivalently a lower bound of $1-\alpha$ on coverage.
\end{proof}

\subsection{Proof of \texorpdfstring{\Cref{thm:diagonal-quadratic}}{the diagonal variance estimator theorem}}
\label{app:proof-diagonal-quadratic}

\begin{proof}
    By \Cref{lma:linearize} and \Cref{asm:linearity}, at $\tilde w = w$ the influence terms are $\psi_i = \frac{2}{p} S_i (2D_i - 1) X_i$ with $X_i = (a_i - \abarn) + (b_i - \bar b) H_i$. The argument of \Cref{thm:variance-limit}, restricted to the terms with $i = j$, gives $\sigmandiag^2 - \frac{1}{n} \sum_i \E[\psi_i^2] \toprob 0$. Since $\psi_i$ and $\psi_j$ are independent when $\dist(i, j) > 2$,
    \[
        \frac{1}{n} \sum_i \E[\psi_i^2] - \sigman^2
        = \frac{1}{n} \sum_i \left(\E[\psi_i]\right)^2
        - \frac{1}{n} \sum_i \sum_{j \in N_2(i), \, j \ne i}
        \left(\E[\psi_i\psi_j] - \E[\psi_i]\E[\psi_j]\right).
    \]

    For the first sum, the calculation in \Cref{app:proof-variance-limit} gives $\E[\psi_i] = (b_i - \bar b)\left(\zeta(1) - \zeta(0)\right)$. This gives the diagonal entries $\Lambda_{ii} = \left(\zeta(1) - \zeta(0)\right)^2$.

    For the off-diagonal raw-product sum, with $i \ne j$ we have
    \[
        X_iX_j = (a_i-\abarn)(a_j-\abarn) +(a_i-\abarn)(b_j-\bar b)H_j +(b_i-\bar b)(a_j-\abarn)H_i +(b_i-\bar b)(b_j-\bar b)H_iH_j.
    \]
    The term in $(a_i - \abarn)(a_j - \abarn)$ vanishes because $S_i, D_i, S_j, D_j$ are jointly independent by \Cref{asm:random} and $\E[2D_i - 1] = 0$. As shown in \Cref{app:proof-efficiency}, the cross term $(a_i - \abarn)(b_j - \bar b)H_j$ and its transpose vanish at $u = 1/2$ because $\E[(2D_i - 1)(2D_j - 1)Z_{ij} \mid S_i = S_j = 1] = (1-2u)/4 = 0$. Only the term in $(b_i - \bar b)(b_j - \bar b)$ survives, so
    \[
        \E[\psi_i \psi_j] = 4 (b_i - \bar b)(b_j - \bar b) \E\big[(2D_i - 1)(2D_j - 1) H_i H_j \mid S_i = S_j = 1\big],
    \]
    and subtracting $\E[\psi_i] \E[\psi_j]$ gives
    \[
        \Lambda_{ij} = \left(\zeta(1) - \zeta(0)\right)^2 - 4\, \E\big[(2D_i - 1)(2D_j - 1) H_i H_j \mid S_i = S_j = 1\big], \qquad i \ne j.
    \]

    For the bound, $\zeta(1) - \zeta(0) = 1 - p/2$ and $H_i$ lie in $[0, 1]$, so $|\Lambda_{ij}| \leqslant \left(\zeta(1) - \zeta(0)\right)^2 + 4 \leqslant 5$.
\end{proof}

\subsection{Efficiency of the true weights}
\label{app:proof-efficiency}

\begin{theorem}[Efficiency of the true weights]
    Suppose the slopes are homogeneous, $b_i \equiv \bar b = \thetan$, and conflicts are resolved symmetrically, $u = \E[U_{ij}] = 1/2$. Then the asymptotic variance of $\thetahat$ is minimized at the true weights $\tilde w = w$.
\label{thm:efficiency}
\end{theorem}

\begin{proof}
    Under homogeneity, $Y_i = a_i + \thetan H_i$, so the unscaled error of \Cref{lma:linearize} becomes
    \[
        \tilde X_i = Y_i - \thetan \tilde H_i - \abarn = (a_i - \abarn) + \bar b \check H_i, \qquad \check H_i \coloneq H_i - \tilde H_i = \sum_k (w_{ik} - \tilde w_{ik}) Z_{ik}.
    \]
    The influence function $\psi_i = \frac{2}{p} S_i (2D_i - 1) \tilde X_i$ therefore splits as $\psi_i = \psi_i^A + \bar b \phi_i$, where
    \[
        \psi_i^A = \frac{2}{p} S_i (2D_i - 1)(a_i - \abarn), \qquad \phi_i = \frac{2}{p} S_i (2D_i - 1) \check H_i,
    \]
    and therefore
    \[
        \sigman^2 = \var\left(\frac{1}{\sqrt n} \sum_i \psi_i\right) = \sigma_A^2 + 2 \bar b C + \bar b^2 V,
    \]
    where $\sigma_A^2 = \var(\frac{1}{\sqrt n} \sum_i \psi_i^A)$ is the no-treatment-effect variance of \Cref{sec:diagonal-variance} and does not depend on weights, $C = \frac{1}{n} \sum_{i, j} \cov(\psi_i^A, \phi_j)$, and $V = \var(\frac{1}{\sqrt n} \sum_i \phi_i)$.

    We first show $C = 0$. For the own-node terms $i = j$, we have $\E[\check H_i \mid S_i = 1, D_i = d] = \zeta(d) \sum_k (w_{ik} - \tilde w_{ik}) = 0$, so $\E[\psi_i^A \phi_i] = \frac{4}{p^2}(a_i - \abarn) \E[S_i \check H_i] = 0$. For $i \ne j$, $\cov(\psi_i^A, \phi_j)$ is nonzero only when $i$ and $j$ are adjacent, in which case it is governed by $\E[S_i (2D_i - 1) S_j (2D_j - 1) Z_{ij}]$, which by \Cref{asm:random} and the assignment rule vanishes at $u = 1/2$, where the fair tie-break leaves $Z_{ij}$ uncorrelated with $(2D_i - 1)(2D_j - 1)$.

    Since $V \geqslant 0$ as a variance and $\check H_i \equiv 0$ when $\tilde w = w$, we have $\sigman^2 = \sigma_A^2 + \bar b^2 V \geqslant \sigma_A^2$, with the minimum attained at the true weights $\tilde w = w$.
\end{proof}

\clearpage
\setcounter{figure}{0}
\renewcommand{\figurename}{Appendix Figure}
\renewcommand{\thefigure}{\thesection.\arabic{figure}}
\renewcommand{\theHfigure}{appendix.\thesection.\arabic{figure}}

\section{Additional results}
\label{app:simulation-results}

\Cref{fig:nonlinear-weighting} illustrates the exposure-level weights in \Cref{thm:nonlinear-weighting} for users with different numbers of equal-weight neighbors.

\begin{figure}[!htbp]
    \centering
    \includegraphics[width=\textwidth]{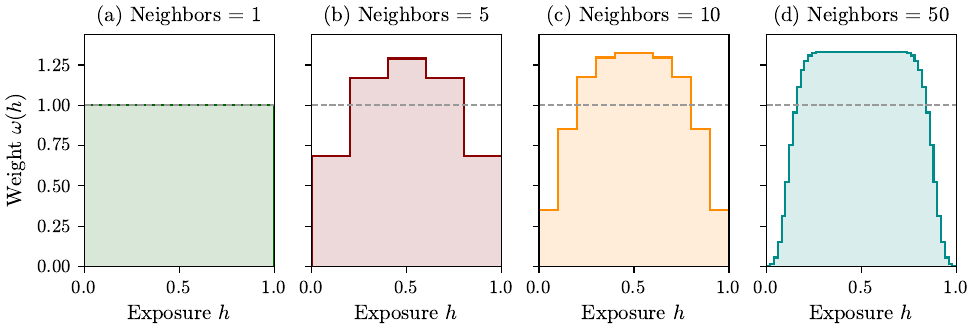}
    \caption{Exposure-level weights under nonlinear dose-response.}
    \label{fig:nonlinear-weighting}
    \vspace{1ex}
    \begin{minipage}{\textwidth}
        \footnotesize \textit{Note:} The panels show the analytical weighting functions for users with 1, 5, 10, and 50 equal-weight neighbors. The experiment sampling proportion is 50\%, and conflicting edges use a symmetric tie-break. The dashed line marks uniform weighting, $\omega_i(h)=1$.
    \end{minipage}
\end{figure}

\Cref{fig:simulation-distributions} reports the sampling distributions of the Wald estimator for the two outcome metrics in our empirical study.

\begin{figure}[!htbp]
    \centering
    \includegraphics[width=\textwidth]{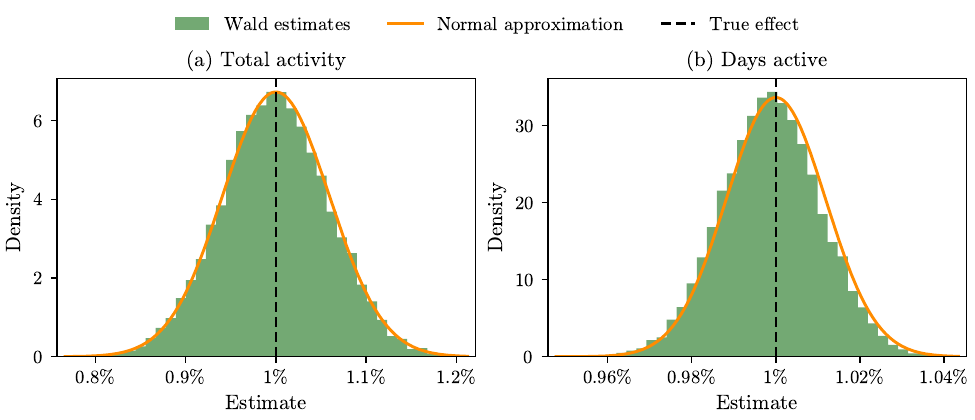}
    \caption{Sampling distributions of the Wald estimator.}
    \label{fig:simulation-distributions}
    \vspace{1ex}
    \begin{minipage}{\textwidth}
        \footnotesize \textit{Note:} The panels show histograms of relative Wald estimates with the correct exposure weights across 20,000 Monte Carlo draws. The normal approximation is the density of a normal distribution centered at the true effect, with standard deviation equal to that of the simulated estimates. The experiment sampling proportion is 50\%, and the simulated relative treatment effect is 1\%.
    \end{minipage}
\end{figure}

\end{appendices}

\end{document}